\documentclass[11pt]{article}
\usepackage[margin=1in]{geometry}
\usepackage[numbers,sort&compress]{natbib}

\usepackage{amssymb}
\usepackage{amsmath}
\usepackage{amsthm}
\usepackage{booktabs}
\usepackage{multirow}
\usepackage{graphicx}
\usepackage[hyphens]{url}
\usepackage{placeins}
\theoremstyle{plain}
\newtheorem{theorem}{Theorem}
\newtheorem{proposition}{Proposition}
\newtheorem{lemma}{Lemma}
\newtheorem{corollary}{Corollary}
\theoremstyle{definition}
\newtheorem{definition}{Definition}
\newtheorem{example}{Example}
\theoremstyle{remark}
\newtheorem{remark}{Remark}

\newcommand{\X}{\mathbf{X}}
\newcommand{\x}{\mathbf{x}}
\newcommand{\z}{\mathbf{z}}
\newcommand{\m}{\mathbf{m}}
\newcommand{\U}{\mathbf{U}}
\newcommand{\mut}{\boldsymbol{\mu}}
\newcommand{\Sig}{\Sigma}
\newcommand{\R}{\mathbb{R}}
\newcommand{\E}{\mathbb{E}}
\newcommand{\Prob}{\mathbb{P}}
\newcommand{\KL}{\mathrm{KL}}
\newcommand{\Unif}{\mathrm{Unif}}
\newcommand{\Ntr}{N_{\mathrm{tr}}}

\newenvironment{procedure}[1]{%
  \begin{center}\begin{minipage}{0.95\linewidth}\rule{\linewidth}{0.8pt}\\[-0.3em]
  \textbf{Procedure #1}\\[-0.6em]\rule{\linewidth}{0.4pt}\par\small}{%
  \par\vspace{-0.3em}\rule{\linewidth}{0.8pt}\end{minipage}\end{center}}

\makeatletter
\newcommand{\appendixsetup}{%
  \renewcommand{\p@section}{Appendix~}%
  \renewcommand{\@seccntformat}[1]{\csname appendixprefix@##1\endcsname\csname the##1\endcsname\quad}%
  \@namedef{appendixprefix@section}{Appendix~}%
  \counterwithin{table}{section}\counterwithin{figure}{section}\counterwithin{equation}{section}}
\makeatother
\usepackage[hidelinks]{hyperref}

\begin{document}

\title{Calibrated Order-Randomized Rosenblatt Tests}
\author{Mehrdad Pournaderi\\[2pt]
{\normalsize Mofid Securities, Iran}\\
{\normalsize\texttt{m.pournaderi@emofid.com}}}
\date{November 2026}
\maketitle

\begin{abstract}
We test whether a multivariate vector $\X$ conforms to a specified distribution $F$, as in copula
modelling and density forecasting. The Rosenblatt transform reduces this to testing uniformity, but
its arbitrary coordinate ordering affects power under dependence. We study
\emph{order randomization}: transforming under many random orderings and merging evidence with
dependence-robust rules. Reordering only redistributes the Mahalanobis signal energy, so one
ordering is a lucky or unlucky draw, and pooling raises calibrated power. We compare against two
order-invariant baselines, the $\chi^2$ energy test and the symmetric-root test. With $F$ known, the
e-value average over orderings is at least as powerful as the symmetric-root test in every setting
we ran, and as the $\chi^2$ test on concentrated and fully diffuse departures. With Gaussian mean
and covariance estimated, the classical conditional probability integral transform keeps each
ordering's $p$-values exact and e-values valid, a parametric bootstrap calibrates the pooled tests,
and the comparisons hold except for deficits of at most $0.02$ against the symmetric-root test on
single-coordinate departures. Validating a Gaussian risk model for nine currencies, the test flags
Brexit, COVID and other episodes; it extends to any null with computable, simulable conditionals.
\end{abstract}


\section{Introduction}\label{sec:intro}

Deciding whether a single multivariate vector $\X=(X_1,\dots,X_n)$ conforms to a specified joint
distribution $F$---the hypothesis $H_0:\X\sim F$---is a basic task in copula modeling,
density-forecast evaluation, simulation-based inference and model checking, and anomaly detection.
It is also the per-hypothesis building block of many multiple-testing problems, where one screens a
large collection of such multivariate hypotheses; our focus is the single test, and the multiplicity
layer (Section~\ref{sec:fdr}) is a standard add-on. A classical and very general device for the
single test is the \emph{Rosenblatt} \cite{rosenblatt1952} transformation, which writes $F$ as a
chain of conditionals and applies the probability integral transform (PIT) to each. Under $H_0$ the
resulting values are i.i.d.\ $\Unif[0,1]$, reducing the question to a test of uniformity addressable
by a univariate combination such as the Simes statistic \cite{simes1986}. This places the problem in
the goodness-of-fit / specification-testing tradition \cite{genest2009,dgt1998}; we stress, however,
that each hypothesis concerns a \emph{single} multivariate observation (the uniforms are its $n$
transformed coordinates, not a sample of repeated draws).

The catch is that to use the transform the analyst must first \emph{order} the coordinates, and that
choice is far from innocuous. Under dependence, different orderings whiten the data differently and
yield tests of materially different power: the same observation, ordered two ways, can comfortably
accept under one and decisively reject under the other (Example~\ref{ex:bivar} makes this concrete). The arbitrariness is well documented, and has
been treated as a nuisance to be removed---by fixing an order by convention, or by symmetrizing a
single statistic over orderings \cite{dgt1998,dovern2020}. We read it instead as an
\emph{opportunity}: if no ordering is canonical, committing to one is a gamble, and it is a gamble
one need not take, because a single observation can be transformed under many orderings at once and
the resulting evidence pooled.

Two questions then organize the paper, and the first is prior to the method. \emph{When should
the Rosenblatt reduction be used at all}, rather than a test that never touches an
ordering---the $\chi^2$ energy test, or a symmetrized whitening (\S\ref{sec:baselines})? Our
answer is a power comparison with these two order-invariant baselines, which fall short in
different ways. The symmetric-root test is strong against concentrated departures and collapses
against diffuse ones of the same energy. The $\chi^2$ energy test has no such weak spot---its
power depends on the departure only through its energy (Proposition~\ref{prop:chi2}) and is
therefore \emph{flat} across shapes---but that common level trails the pooled tests against
concentrated departures and is competitive only against diffuse ones (\S\ref{sec:headline}). At
known $F$, the e-value average over orderings is at least as powerful as the symmetric-root test
in every setting we ran, and at least as powerful as the $\chi^2$ test against both concentrated
and fully diffuse departures. With $\Sig$ estimated and the exact whitening of
\S\ref{sec:calib}, the same holds from the smallest training sample we ran, four observations per
dimension, except against a single shifted coordinate, where the symmetric-root test keeps an edge
of at most $0.02$ until the training sample reaches about ten observations per dimension. The geometry behind this is the energy
invariance of Lemma~\ref{lem:energy}: reordering conserves the total Mahalanobis signal energy
and only redistributes it across the whitened coordinates, so a fixed ordering's fortunes
depend entirely on how that redistribution happens to fall---committing to one ordering is a
gamble. The second question is then the practical one this paper answers in full:

\begin{quote}\itshape
Does applying the Rosenblatt transform under many random orderings, and combining the resulting
tests, beat using a single (arbitrarily chosen) ordering?
\end{quote}

Whether pooling actually helps is far from obvious, and answering it is most of this paper. Three
things must line up. \emph{First}, one must understand what reordering does to a departure from $F$:
we show (Section~\ref{sec:theory}) that it conserves the total Mahalanobis signal energy and merely
\emph{redistributes} it across the whitened coordinates, so every ordering is a different view of one
fixed signal, a single ordering is a lucky or unlucky draw, and pooling averages over the luck---by
an amount that, in every design we ran, grows with dependence. \emph{Second}, the per-coordinate test must look in both tails, because the whitening routinely
sends a departure into either tail depending on the ordering (Section~\ref{sec:misspec}). \emph{Third}, and this is the crux, the pooled test must be \emph{calibrated}: naive
pooling pays multiplicity penalties that almost exactly cancel its advantage, and under an estimated
null it is not even valid with the plug-in whitening. Validity has a classical exact remedy,
which we adopt: whitening each ordering by the conditional probability integral transform
\cite{oreilly1973,rincon1979} makes its scores exactly standard normal under the estimated null
(\S\ref{sec:calib}). The multiplicity penalties remain, so a re-estimating parametric bootstrap is
needed---and in restoring level it \emph{unveils} the very power gain those penalties had hidden
(Section~\ref{sec:exp}).

When the three align, the payoff is real. Its most useful form is the comparison with the
order-invariant baselines: the e-value average over orderings is in the top two against both
concentrated and fully diffuse departures, whereas the symmetric-root test, strong against
concentrated departures, collapses against diffuse ones of the same energy, where its power falls
by more than $0.30$, and the $\chi^2$ test, whose power is the same for every shape at fixed energy,
is well below the pooled tests against concentrated departures and competitive only against
diffuse ones. Against a single random ordering---the expected
power of an arbitrarily chosen one---the improvement is smaller but systematic: $0.03$--$0.06$ for
sparse departures, growing with dependence in every design we ran, in the direction the
Jensen-gap mechanism of Section~\ref{sec:theory} suggests (from $+0.02$ at
$\rho=0.2$ to $+0.06$ at $\rho=0.8$, known covariance), up to $+0.17$ for dense ones, and, for
sparse departures, of the same size when the covariance must be estimated ($+0.04$ to $+0.06$ for the e-value average at $\Ntr=4n$). The geometry behind this is provable, and Section~\ref{sec:theory} does so: the log-evidence
advantage of order-averaging is exactly the spread of views across orderings. We identify the
combiners to use and the number of orderings, and mark where pooling is merely neutral
(near-independence). Screening many such hypotheses at once needs no new machinery, only
BH or e-BH on top (Section~\ref{sec:fdr}); it is also where pooling pays most against the order-invariant baselines, since the
hypotheses in a screen need not deviate from the null distribution in the same way. That is also
where a modern alternative meets us head-on, and \S\ref{sec:conformal} compares the procedure with
conformal novelty detection \cite{bates2023,marandon2024}: the two test different nulls---conformity
to a specified $F$ against exchangeability with a reference sample---and when the model is
correct, the one setting in which their power can be compared, the calibrated pooled test leads in
the sparse regime that screening actually presents, across reference samples up to ten
times the size of the screen. Section~\ref{sec:fx} takes the whole pipeline
to real data: validating a
yearly-refitted Gaussian risk model on nine currencies over 2016--2025, where the calibrated test
rejects at close to the nominal rate in calm years and the pooled e-values single out the Brexit, COVID, 2022, and 2025
episodes. Section~\ref{sec:prelim} first assembles the tools; readers fluent in the
Rosenblatt transform, Simes, and e-values may skim it.

\paragraph{Notation}
$\Phi$ and $\varphi$ are the standard normal CDF and density; $z_a=\Phi^{-1}(1-a)$ the upper-$a$
quantile. $S_n$ is the symmetric group (the set of all $n!$ orderings) on $\{1,\dots,n\}$. For
$\sigma\in S_n$, $P_\sigma$ is the permutation matrix with $(P_\sigma\x)_k=x_{\sigma(k)}$. For a
vector $\mathbf p\in[0,1]^n$, $p_{(1)}\le\dots\le p_{(n)}$ are its order statistics (the entries
sorted increasingly). ``i.i.d.'' is independent and identically distributed. All distributions are
absolutely continuous unless stated.

\section{Preliminaries}\label{sec:prelim}

Pooling orderings needs only three ingredients: the transform that turns a multivariate hypothesis
into a uniformity question, its explicit Gaussian form (a Cholesky whitening), and the univariate
rules we pool with. We collect them here together with the validity facts the later arguments rely
on; the reader who knows them can move to Section~\ref{sec:method}.

\subsection{The Rosenblatt transformation}

Let $\X\sim F$ on $\R^n$ with absolutely continuous law, and fix an ordering $\sigma\in S_n$. Write
$F_k^\sigma(\cdot\mid\cdot)$ for the conditional CDF of $X_{\sigma(k)}$ given
$X_{\sigma(1)},\dots,X_{\sigma(k-1)}$ under $F$ (with $F_1^\sigma$ the marginal CDF of
$X_{\sigma(1)}$). The \emph{Rosenblatt map} $R_\sigma:\R^n\to[0,1]^n$ is
\begin{equation}\label{eq:rosen}
\big(R_\sigma(\x)\big)_k \;=\; U_k \;=\; F_k^\sigma\!\big(x_{\sigma(k)}\mid x_{\sigma(1)},\dots,x_{\sigma(k-1)}\big),
\qquad k=1,\dots,n.
\end{equation}

\begin{theorem}[Rosenblatt, 1952]\label{thm:rosen}
If $\X\sim F$ then $\U=R_\sigma(\X)$ has i.i.d.\ $\Unif[0,1]$ coordinates, for every fixed
$\sigma\in S_n$.
\end{theorem}
\begin{proof}[Proof sketch]
By the probability integral transform, conditionally on the prefix
$X_{\sigma(1)}$, \dots, $X_{\sigma(k-1)}$ the variable
$U_k=F_k^\sigma(X_{\sigma(k)}\mid\cdot)$ is $\Unif[0,1]$ and hence independent of the prefix.
Therefore $U_k$ is independent of $U_1,\dots,U_{k-1}$ (which are functions of the prefix) and
uniform; induction over $k$ gives joint independence and uniformity.
\end{proof}

In words: take the coordinates one at a time in the chosen order; for each, ask ``given everything
seen so far, at which quantile of its predicted conditional distribution did this coordinate land?''.
If the model is right, those quantile positions are pure noise---independent uniforms; any structure
in them is evidence against the model. Theorem~\ref{thm:rosen} is the engine: it converts
``$\X\sim F$?'' into ``are the $U_k$ i.i.d.\ uniform?''. The map $R_\sigma$ is a measurable bijection
(a triangular, Knothe--Rosenblatt map) whenever the conditionals are strictly increasing, a fact we
use in Section~\ref{sec:misspec}.

\subsection{The Gaussian case: Cholesky whitening}\label{sec:gaussrosen}

Let $F=\mathcal N(\mathbf 0,\Sig)$. Let $\Sig_\sigma=P_\sigma\Sig P_\sigma^\top$ be the permuted
covariance and $L_\sigma$ its lower-triangular Cholesky factor, $\Sig_\sigma=L_\sigma L_\sigma^\top$.
Define the whitened scores
\begin{equation}\label{eq:whiten}
\z^\sigma \;=\; L_\sigma^{-1}\,P_\sigma\X \;=\; L_\sigma^{-1}\X_\sigma .
\end{equation}
Because $L_\sigma$ is lower triangular, the $k$-th entry $z^\sigma_k$ is an affine function of
$x_{\sigma(1)},\dots,x_{\sigma(k)}$, and a direct computation identifies it with the standardized
Gaussian conditional:
\[
z^\sigma_k=\frac{x_{\sigma(k)}-\mu_{k\mid<k}(\x)}{s_{k\mid<k}},\qquad
F_k^\sigma\!\big(x_{\sigma(k)}\mid x_{\sigma(1)},\dots,x_{\sigma(k-1)}\big)=\Phi(z^\sigma_k),
\]
where $\mu_{k\mid<k}$ and $s_{k\mid<k}$ are the conditional mean and standard deviation of
$X_{\sigma(k)}$ given the prefix. Hence the Gaussian Rosenblatt map is
$U^\sigma_k=\Phi(z^\sigma_k)$, and under $H_0$, $\z^\sigma\sim\mathcal N(\mathbf 0,I)$. We work
directly with the scores $\z^\sigma$ (equivalently the normal scores $\Phi^{-1}(U^\sigma_k)=z^\sigma_k$).
In plain terms, \emph{whitening} means transforming correlated coordinates into standardized,
uncorrelated ones: each score $z^\sigma_k$ answers ``by how many (conditional) standard deviations
did coordinate $\sigma(k)$ miss the value its predecessors predicted for it?''---so under the model
every score should look like a fresh standard normal draw.

\begin{example}[Why the ordering matters]\label{ex:bivar}
Take $n=2$, $\Sig=\left(\begin{smallmatrix}1&\rho\\\rho&1\end{smallmatrix}\right)$ with $\rho=0.9$, and
observe the single point $\x=(2,0)$. The two orderings give:
\[
\begin{aligned}
\text{order }(1,2):\quad
&\z=\Big(x_1,\ \tfrac{x_2-\rho x_1}{\sqrt{1-\rho^2}}\Big)=(2.00,\,-4.13),\\
&\U=\Phi(\z)=(0.977,\ 1.8\times10^{-5});\\
\text{order }(2,1):\quad
&\z=\Big(x_2,\ \tfrac{x_1-\rho x_2}{\sqrt{1-\rho^2}}\Big)=(0.00,\,+4.59),\\
&\U=\Phi(\z)=(0.500,\ 0.999998).
\end{aligned}
\]
The same observation produces \emph{entirely different} transforms---and not merely a relabeling: the
multisets $\{0.977,\,1.8\times10^{-5}\}$ and $\{0.5,\,0.999998\}$ differ, so even a
permutation-symmetric combination such as Simes depends on the ordering.

\emph{Whitening view.} Both transforms apply the same triangular factor
\[
L^{-1}=\begin{pmatrix}1&0\\[2pt]-\rho/\sqrt{1-\rho^2}&1/\sqrt{1-\rho^2}\end{pmatrix}
\]
to the \emph{ordered} data $\X_\sigma$---here $\Sig$ is exchangeable, so $L_\sigma$ is unchanged and only
$\X_\sigma$ is permuted. Triangularity makes the slot decisive (slot one passes through; slot two is the
residual after regressing it out), so the swap exchanges the roles $X_2\mid X_1$ and $X_1\mid X_2$. (For
non-exchangeable $\Sig$, $\Sig_\sigma$ and hence $L_\sigma$ change with the ordering too.)

The effect is interpretable.
Under strong positive correlation, $\x=(2,0)$ is surprising because $X_2$ ``should'' track $X_1$.
Order $(1,2)$ charges that surprise to the conditional $X_2\mid X_1$ (observed $0$ versus expected
$\rho x_1=1.8$), producing a tiny $U_2$; order $(2,1)$ charges it to $X_1\mid X_2$ (observed $2$
versus expected $0$), producing a $U_2$ near one. The surprise has the \emph{same magnitude}---indeed
$\|\z\|^2=\x^\top\Sig^{-1}\x=21.05$ under both orderings (Lemma~\ref{lem:energy})---but it surfaces in
\emph{opposite tails}. A one-sided test that flags small uniforms detects it under order $(1,2)$ yet
is blind to it under order $(2,1)$; a two-sided test catches both. This single point already exhibits
the paper's two levers: reordering reshuffles a fixed amount of signal across coordinates and tails
(the ``orbit'' of Section~\ref{sec:theory}), and the base test must be two-sided
(Section~\ref{sec:misspec}) to see it whichever way it lands.
\end{example}

\begin{remark}[Why the \emph{triangular} factor]\label{rem:tri}
The whitening matrix is not unique: any $W$ with $W\Sig_\sigma W^\top=I$ yields i.i.d.\ scores, and two
such factors differ by an orthogonal rotation. Rosenblatt singles out the lower-triangular Cholesky
factor (unique with positive diagonal) precisely because triangularity \emph{is} the chain of
conditionals: $z^\sigma_k$ standardizes $X_{\sigma(k)}$ given its predecessors. This is what makes
ordering matter---Cholesky does not commute with permutation, $\mathrm{chol}(P_\sigma\Sig P_\sigma^\top)\neq
P_\sigma\,\mathrm{chol}(\Sig)$. A symmetric root gives instead the whitening matrix
$W_\sigma=\Sig^{-1/2}P_\sigma^\top$ (its inverse $P_\sigma\Sig^{1/2}$ is the corresponding
\emph{non-triangular} covariance factor of $\Sig_\sigma$), whose scores
$W_\sigma\X_\sigma=\Sig^{-1/2}\X$ are \emph{order-invariant}---a
legitimate test, but a global mix rather than a chain of conditionals; for general
non-Gaussian $F$ a linear covariance-root transformation does not produce independent uniform
coordinates, and the nonlinear alternatives (optimal-transport maps to a reference law) require
additional construction, whereas the sequential conditional CDFs are always available. Order-invariance is no free lunch here: by
Lemma~\ref{lem:energy} the symmetric whitening is a single fixed profile on the energy sphere---one
point, not the orbit that pooling exploits.
\end{remark}

\subsection{Combining uniforms: Simes, e-values, BH and e-BH}

\begin{definition}[$p$-value, e-value]
A statistic $T\ge 0$ is a $p$-\emph{value} if $\Prob_{H_0}(T\le t)\le t$ for all $t\in[0,1]$, and an
$e$-\emph{value} if $\E_{H_0}[T]\le 1$. By Markov's inequality, if $E$ is an e-value then
$E\ge 1/\alpha$ is a level-$\alpha$ test, and $1/E$ (capped at $1$) is a $p$-value.
\end{definition}

\begin{definition}[Simes statistic]
For per-coordinate $p$-values $\mathbf p=(p_1,\dots,p_n)$, the Simes combination is
$S(\mathbf p)=\min_{1\le i\le n} n\,p_{(i)}/i$. If the $p_i$ are i.i.d.\ $\Unif[0,1]$ then
$S(\mathbf p)\sim\Unif[0,1]$ \cite{simes1986}; thus $S$ is itself a $p$-value under independence,
exactly the situation produced by Theorem~\ref{thm:rosen}.
\end{definition}

Simes scans all sparsity levels at once: the $i=1$ term $n\,p_{(1)}$ asks ``is the single most
extreme coordinate surprising after a multiplicity correction?'', the $i=n$ term asks ``are all
coordinates a bit off?'', and the minimum takes the strongest of these views---so one statistic is
sensitive both to one large departure and to many small ones.

\begin{lemma}[Merging rules]\label{lem:merge}
Let $E_1,\dots,E_M$ be e-values and $p_1,\dots,p_M$ be $p$-values, with \emph{arbitrary} dependence.
Then: (i) the average $\frac1M\sum_m E_m$ is an e-value \cite{vovkwang2021}; (ii) twice the average
$\min(1,\frac2M\sum_m p_m)$ is a $p$-value (arithmetic-mean merging, \cite{vovkwang2020}); and
(iii) $\min(1, M\min_m p_m)$ is a $p$-value (Bonferroni / union bound).
\end{lemma}
\begin{proof}
(i) $\E_{H_0}[\frac1M\sum_m E_m]=\frac1M\sum_m\E_{H_0}[E_m]\le 1$ by linearity, with no independence
needed. (iii) $\Prob_{H_0}(\min_m p_m\le \alpha/M)\le\sum_m\Prob(p_m\le\alpha/M)\le M\cdot\alpha/M=\alpha$.
(ii) is the arithmetic-mean member ($r=1$) of the generalized-mean merging family of
\cite{vovkwang2020}, where it is shown that twice the average of arbitrarily dependent $p$-values is
again a $p$-value, and that the factor $2$ cannot be improved in general.
\end{proof}

Finally, when \emph{many} hypotheses are tested at once (Sections~\ref{sec:fdr} and~\ref{sec:fx}
screen thousands of days), controlling the per-test false-alarm rate is not enough: with $N=1{,}000$
true nulls at $\alpha=0.05$ one expects $50$ false alarms. The standard remedy controls the
\emph{false discovery rate} (FDR): the expected fraction of the rejected hypotheses that are false
alarms.

\begin{definition}[BH and e-BH]
Given $p$-values $p_1,\dots,p_N$ for $N$ hypotheses and target FDR $q$, the Benjamini--Hochberg
procedure \cite{bh1995} rejects the $k^\ast$ hypotheses with smallest $p$, where
$k^\ast=\max\{k: p_{(k)}\le qk/N\}$; it controls FDR under independence or positive dependence
(PRDS). Given e-values $E_1,\dots,E_N$, the \emph{e-BH} procedure \cite{wangramdas2022} rejects the
$k^\ast$ with largest $E$, where $k^\ast=\max\{k: E_{(k)}\ge N/(qk)\}$ ($E_{(1)}\ge\dots\ge E_{(N)}$);
it controls FDR at $q$ under \emph{arbitrary} dependence among the $E_i$, with no correction.
\end{definition}

\subsection{Order-invariant baselines}\label{sec:baselines}

The ordering problem has an obvious response: do not order. Two standard tests take it, and they
are the references against which everything in this paper is measured, so we fix them here. Both
are stated for the Gaussian null $\mathcal N(\mathbf 0,\Sig)$ of \S\ref{sec:gaussrosen}; both are
invariant to relabelling the coordinates; and both are carried through every experiment below,
bootstrap-calibrated by the same Procedure~1 as the Rosenblatt-based tests, so that all comparisons
are made at matched size---exactly at known $F$, and approximately under estimation, where
\S\ref{sec:calib-exp} reports the realized sizes.

\begin{definition}[$\chi^2$ energy test]\label{def:chi2}
Reject for large values of the quadratic form
\begin{equation}\label{eq:chi2}
T_{\chi^2}(\X) \;=\; \X^\top\Sig^{-1}\X \;\sim\; \chi^2_n \quad\text{under } H_0 .
\end{equation}
This is the Mahalanobis norm of $\X$, the classical omnibus test against mean shifts of a Gaussian null. It is order
invariant because the quadratic form does not see the labels at all.
\end{definition}

Under a mean shift, $\X\sim\mathcal N(\mut,\Sig)$, the statistic is noncentral $\chi^2_n$ with
noncentrality parameter $\mut^\top\Sig^{-1}\mut$, the squared Mahalanobis distance of the
alternative from the null. We write $\mathrm{ncp}:=\mut^\top\Sig^{-1}\mut$ and use it throughout to
set the strength of a departure; it is also the total signal energy of Lemma~\ref{lem:energy},
which reordering conserves.

\begin{definition}[Symmetric-root test]\label{def:sym}
Whiten with the symmetric square root of Remark~\ref{rem:tri}, $\mathbf z=\Sig^{-1/2}\X$, and apply
the two-sided Simes statistic of \S\ref{sec:base} to the per-coordinate $p$-values
$p_j=2\Phi(-|z_j|)$. Order invariance holds because the symmetric root commutes with permutation,
$(P_\sigma\Sig P_\sigma^\top)^{-1/2}=P_\sigma\Sig^{-1/2}P_\sigma^\top$, which is exactly what the
triangular Cholesky factor fails to do.
\end{definition}

The two differ in what they are sensitive to, and this will matter throughout. The $\chi^2$ test
aggregates the whitened coordinates by summing their squares, so it spreads its $n$ degrees of
freedom evenly: its power is the same for every shape of departure at fixed energy, which makes
it competitive only where the shape-sensitive tests do badly, on diffuse departures; the symmetric-root test aggregates by
a min-type rule, so it is at its best when a departure is concentrated in few coordinates. Neither
is uniformly preferable, and \S\ref{sec:headline} shows that the ranking between them inverts
between the two regimes.

There is a second, sharper reason to keep $T_{\chi^2}$ in view. By Lemma~\ref{lem:energy} the total
Mahalanobis signal energy is conserved under reordering, and \eqref{eq:chi2} is a function of that
energy alone. Proposition~\ref{prop:chi2} makes the consequence precise: the $\chi^2$ test is
\emph{orbit-flat}, its power identical under every ordering. That is simultaneously its guarantee
and its ceiling, and it is the cleanest way to see what order randomization can and cannot buy---a
test that already ignores the ordering has nothing to gain from averaging over it.

Symmetrizing a Rosenblatt-based statistic to remove the ordering is the route taken by Dovern and
Manner \cite{dovern2020} for multivariate density-forecast calibration. The present paper reaches
the same symmetry by a different road: rather than removing the ordering analytically, it draws
many orderings and pools the resulting evidence---whose $M\to\infty$ average is itself a
symmetrization (Section~\ref{sec:disc})---so the difference is \emph{what} is aggregated and
with what finite-$M$ guarantee, not aggregation versus symmetrization.

\section{Method: an order-randomized test for one multivariate hypothesis}\label{sec:method}

The construction is short: whiten $\X$ under each of $M$ random orderings, score each whitening, and
pool the scores. Two design choices in it are not arbitrary but are the two things that, by the
introduction's account, have to be right for pooling to pay off. We state them here and foreshadow
their justification: the base statistic is \emph{two-sided} (Section~\ref{sec:misspec} shows a
one-sided one is blind to half the signal), and the pooled test is \emph{bootstrap-calibrated}
(Section~\ref{sec:exp} shows this is  the step that uncovers the power
gain; under an estimated null the whitening is also replaced by its exact predictive
counterpart, \S\ref{sec:calib}, without which the test would not even be valid). With those flagged, everything below is routine.

\subsection{Two-sided base statistics}\label{sec:base}

Given the scores $\z^\sigma\sim\mathcal N(\mathbf 0,I)$ under $H_0$, we form a base statistic that is
sensitive to departures in \emph{either} tail (the reason is the directional blind spot of
Section~\ref{sec:misspec}). The toolkit has three base statistics, differing in the signal
\emph{shape} they are tuned to: Simes and the mixture e-value for concentrated (sparse) departures,
Fisher for diffuse (dense) ones.

\paragraph{Two-sided Simes}
Per-coordinate two-sided $p$-values $p^\sigma_i=2\,\Phi(-|z^\sigma_i|)$, combined by
$S(\mathbf p^\sigma)=\min_i n\,p^\sigma_{(i)}/i$. Under $H_0$ this is a $p$-value.

\paragraph{Two-sided mixture e-value}
With a small grid $\mathcal T\subset(0,\infty)$ of ``bet sizes,''
\begin{equation}\label{eq:evalue}
E^\sigma=\frac1n\sum_{k=1}^n\frac1{|\mathcal T|}\sum_{\tau\in\mathcal T} e_\tau(z^\sigma_k),
\qquad e_\tau(z):=e^{-\tau^2/2}\cosh(\tau z).
\end{equation}

\begin{lemma}[Validity of the mixture e-value]\label{lem:evalue}
$\E_{\mathcal N(0,1)}[e_\tau(Z)]=1$ for every $\tau$, and hence $\E_{H_0}[E^\sigma]=1$.
\end{lemma}
\begin{proof}
$\E[e_\tau(Z)]=e^{-\tau^2/2}\,\E[\cosh(\tau Z)]=e^{-\tau^2/2}\cdot\tfrac12(\E[e^{\tau Z}]+\E[e^{-\tau Z}])
=e^{-\tau^2/2}\cdot\tfrac12(e^{\tau^2/2}+e^{\tau^2/2})=1$, using the Gaussian MGF
$\E[e^{tZ}]=e^{t^2/2}$. Averaging unit-mean terms over $k$ and $\tau$ preserves unit mean.
\end{proof}

One caveat must be stated plainly, because everything in Sections~\ref{sec:calib}
and~\ref{sec:fdr} turns on it. Lemma~\ref{lem:evalue} requires the scores to be
\emph{exactly} standard normal, which they are when the null---the distribution being
tested---is fully specified, as when one tests an issued forecast (Section~\ref{sec:fx}).
It is \emph{not} preserved by plugging in an estimated covariance: if the null is an unknown
Gaussian population whose $\Sig$ is estimated from a reference sample, the plug-in scores are
heavier-tailed than normal and $E^\sigma$ need not be an e-value at all. The one-dimensional
case makes the failure vivid: with an independent unregularized variance estimate,
$Z=X/\hat\sigma$ is Student-$t$, and $\E[e^{-\tau^2/2}\cosh(\tau Z)]=\infty$ for every
$\tau>0$, since the $t$ distribution has no moment generating function. The remedy is to correct the
scores rather than give up the e-value. The exact predictive whitening of \S\ref{sec:calib} turns
the plug-in scores into scores that are exactly standard normal under the estimated-population
null defined there, so Lemma~\ref{lem:evalue} applies to them as stated, and the e-value reading,
with everything that relies on it (Markov thresholds, e-BH), is available under that null too. On plug-in
scores the mixture statistic could serve only as a test statistic, calibrated into a $p$-value
by Procedure~1.

The statistic $e_\tau$ is the likelihood ratio of $\mathcal N(0,1)$ against the symmetric mixture
$\tfrac12\mathcal N(\tau,1)+\tfrac12\mathcal N(-\tau,1)$; \eqref{eq:evalue} bets on a sparse,
sign-agnostic mean departure. The ``bet size'' $\tau$ sets the stakes: $e_\tau(z)$ is
increasing in $|z|$ for every $\tau$, growing like $e^{\tau|z|}$, but it starts from
$e_\tau(0)=e^{-\tau^2/2}$, so a large $\tau$ pays off faster on a large departure and loses more
on a null coordinate ($\E_{\mathcal N(0,1)}\log e_\tau(Z)\approx-0.13$, $-0.94$, $-2.7$ at
$\tau=1,2,3$). Averaging over the small grid $\mathcal T=\{1,2,3\}$ hedges between the two without
tuning to an unknown truth. The coordinate \emph{mixture} (a sum over $k$, not a product) is
the other design choice: a product $\prod_k e_\tau(z^\sigma_k)$ pays that ``betting tax'' on every
noise coordinate, and under a sparse alternative the few signal coordinates must repay $n-1$ of
them; for moderate departures and the larger bet sizes they do not (at $n=10$ and
$\mathrm{ncp}=12$, in the ordering that whitens the shifted coordinate last so that the whitened
profile is $(0,\dots,0,\sqrt{12})$, the $\tau=3$ product has expected log-value
$5.2-9\times2.7<0$, and it is negative under every ordering of that design), whereas in the
mixture a noise coordinate merely dilutes the evidence.

\paragraph{Two-sided Fisher}
Fisher's combination on the two-sided $p$-values,
\[
p^\sigma_F \;=\; \bar F_{\chi^2_{2n}}\!\Big(-2\sum_{k=1}^n \log p^\sigma_k\Big),
\]
where $\bar F_{\chi^2_{2n}}$ is the $\chi^2_{2n}$ survival function. Under $H_0$ the $p^\sigma_k$
are i.i.d.\ uniform, so $p^\sigma_F$ is an exact $p$-value.

\begin{remark}[Where Fisher fits---and where it cannot go]\label{rem:fisher}
Fisher aggregates \emph{every} coordinate, which makes it the dense-friendly member of the toolkit:
strong when the departure is spread across many coordinates, weak when it is concentrated
(empirically the weakest method in the study against sparse alternatives, behind even the $\chi^2$
energy test, while for dense alternatives it is competitive; Figure~\ref{fig:power}).
On \emph{one-sided} $p$-values it would inherit the directional blind spot of
Section~\ref{sec:misspec} in full, so the two-sided form above is essential. One boundary is
firm: Fisher's $\chi^2$ null law requires \emph{independent} $p$-values, so with that reference law it may only be
used \emph{within} an ordering---never to merge across orderings, where the per-ordering statistics
are strongly dependent recomputations of one observation; calibrated against its joint null
distribution by Procedure~1, the across-ordering Fisher statistic would be valid like any other,
but we do not pursue it. The Fisher $p$-values
$p^{\sigma_m}_F$ can be merged across orderings with the dependence-robust rules of
Lemma~\ref{lem:merge}, and we ran both aggregation bases through every pooled combiner of
\S\ref{sec:combine} in every configuration of \S\ref{sec:exp}. The outcome is one-sided enough to
settle the design. Pooling adds almost nothing to Fisher
itself---at known $F$ the calibrated arithmetic-mean merge over $M=12$ orderings matches single-ordering Fisher
to within $0.005$, and the Bonferroni merge gains at most $+0.02$ (at $\Ntr=4n$, at most
$0.017$ and $+0.05$)---so the Fisher-based combiners
inherit Fisher's shape profile: against every sparse departure they trail the corresponding
Simes-based combiner by $0.04$--$0.29$, and on the dense departure, Fisher's home ground, the
Simes-based Bonferroni still leads at known $F$ ($0.691$ vs.\ $0.642$) while the Fisher-based
$p$-merge edges its Simes counterpart ($0.623$ vs.\ $0.577$); under heavy estimation ($\Ntr=4n$) both Simes-based combiners lead there
($0.540$ vs.\ $0.533$ and $0.605$ vs.\ $0.510$). We therefore
aggregate only the Simes base and report Fisher as a single-ordering test. The mechanism is the
orbit theory. A sum over all coordinates is only weakly profile-sensitive---much closer to the
orbit-flat $\chi^2$ of Proposition~\ref{prop:chi2} than to the min-type Simes---so its per-ordering
values barely vary, and there is
nothing for pooling to average over---the same mechanism that makes the Jensen gap of
Proposition~\ref{prop:main} small for a profile-insensitive e-value, offered here as an
interpretation: the proposition concerns e-value averaging, whereas the Fisher $p$-values are
merged through their mean or minimum, so it does not itself deliver this conclusion. The pooling gain is a property of \emph{profile-sensitive} base
statistics, and Fisher is the demonstration of the converse.
\end{remark}

\subsection{Combining across orderings}\label{sec:combine}

Draw $\sigma_1,\dots,\sigma_M$ i.i.d.\ uniform on $S_n$, independent of the data. The per-ordering
statistics are recomputations on the \emph{same} $\X$ and are therefore strongly dependent;
Lemma~\ref{lem:merge} lets us combine them regardless. Write $P^\sigma$ for a generic per-ordering
base $p$-value---either the two-sided Simes $S(\mathbf p^\sigma)$ or the two-sided Fisher
$p^\sigma_F$ of \S\ref{sec:base}. We consider:
\begin{itemize}
\item \textbf{e-value averaging:} $\bar E_M=\frac1M\sum_{m=1}^M E^{\sigma_m}$; reject if
$\bar E_M\ge 1/\alpha$.
\item \textbf{$p$-merging:} $P^{\mathrm{merge}}_M=\min(1,\frac2M\sum_{m}P^{\sigma_m})$.
\item \textbf{Bonferroni-over-orders:} $P^{\mathrm{Bonf}}_M=\min(1,M\min_m P^{\sigma_m})$.
\end{itemize}
The capped forms are the nominal merged $p$-values. Whenever a combiner is
\emph{calibrated} (Procedure~1 below), the statistic
handed to the calibration is the uncapped summary---$\frac1M\sum_m P^{\sigma_m}$,
$\min_m P^{\sigma_m}$, or $\bar E_M$---of which the capped $p$-value is a monotone function below
the cap. The distinction matters under a nonrandomized rank rule: the cap places an atom at $1$, an
observation on the atom can never reject, and the calibrated test's rejection probability is
bounded by the null mass below the cap once that mass falls under $\alpha$ (with independent
coordinates and known $\Sig$ every ordering returns the same Simes value $P$, so at $M=128$
the capped Bonferroni statistic equals $1$ with probability $127/128$ and the test could reject
at most $0.8\%$ of the time). Capping is therefore reserved for reporting a nominal merged
$p$-value. (Every experiment calibrates the uncapped summaries. The
screening experiments of Section~\ref{sec:fdr} were first run on the capped forms, and calibrating
the uncapped summaries on the same draws reproduces every one of their BH decisions.)
By Lemma~\ref{lem:merge}, $\bar E_M$ is an e-value and $P^{\mathrm{merge}}_M,P^{\mathrm{Bonf}}_M$ are
$p$-values under $H_0$ for \emph{every} $M$, any dependence among orderings, and either choice of
base $p$-value; randomizing the orderings independently of the data preserves these properties
(condition on $\sigma_{1:M}$ and average). The $M\to\infty$ limit of $\bar E_M$ is the
order-invariant symmetrized e-value $\E_\sigma[E^\sigma(\X)\mid\X]$; finite $M$ is a Monte-Carlo
approximation. We evaluated the pooled combiners on both bases; the Simes base won in every
sparse configuration at both training sizes, and on the dense departure in every case but the $p$-merge at known $F$ (Remark~\ref{rem:fisher}), so
the experiments report Simes-based pooling and Fisher as a single-ordering test.

\subsection{Calibration under an estimated null}\label{sec:calib}

In practice $\Sig$ is unknown and estimated as $\hat\Sig$ from a reference sample of size $\Ntr$.
The plug-in scores $\hat\z^\sigma=\hat L_\sigma^{-1}\X_\sigma$ are then only approximately
$\mathcal N(\mathbf 0,I)$ under $H_0$---heavier-tailed, and increasingly so along the
ordering, since the $k$-th score is scaled by a residual variance estimated after fitting $k-1$
regression coefficients, so the nominal thresholds of Section~\ref{sec:combine} are no
longer exact and the test is miscalibrated (Section~\ref{sec:exp}). Two steps follow:
an exact version of the whitening makes every ordering's scores exactly standard normal, which
makes every test valid, and a re-estimating parametric bootstrap then brings the conservative
pooled combiners approximately to the nominal level.

\paragraph{Exact predictive whitening}
The null must first be stated precisely. Under the \emph{estimated-population null}, the test
vector $\X$ and the reference sample $\mathbf Y_1,\dots,\mathbf Y_N$ ($N=\Ntr$) are independent draws
from one Gaussian $\mathcal N(\mut,\Sig)$, with $\Sig$ unknown and the mean either known
($\mut=\mathbf 0$ in our simulations) or estimated as well (as in Section~\ref{sec:fx}). Let
the centering point $\mut_0$ be the known mean, $\mut_0=\mut$, when the mean is known, and the
reference mean $\bar{\mathbf Y}$ when it is estimated; let
$A=\sum_{i=1}^N(\mathbf Y_i-\mut_0)(\mathbf Y_i-\mut_0)^\top$, and, for an
ordering $\sigma$,
$\mathbf u^\sigma=\mathrm{chol}(P_\sigma AP_\sigma^\top)^{-1}P_\sigma(\X-\mut_0)$; with
$\hat\Sig=A/N$ this is the plug-in score divided by $\sqrt N$. Define, for $k=1,\dots,n$,
\begin{equation}\label{eq:exactz}
t^\sigma_k=\frac{\sqrt{\nu_k}\,u^\sigma_k}{\sqrt{c+\sum_{j<k}(u^\sigma_j)^2}},\qquad
\tilde z^\sigma_k=\Phi^{-1}\bigl(F_{\nu_k}(t^\sigma_k)\bigr),
\end{equation}
where $F_\nu$ is the Student-$t$ distribution function with $\nu$ degrees of freedom, and
$(\nu_k,c)=(N-k+1,\,1)$ when the mean is known and $(N-k,\,1+1/N)$ when it is estimated.

\begin{lemma}[Exact predictive whitening]\label{lem:cpit}
Under the estimated-population null, with $N\ge n$ when the mean is known and $N>n$ when it is
estimated, $\tilde z^\sigma_1,\dots,\tilde z^\sigma_n$ are i.i.d.\ $\mathcal N(0,1)$ for every
ordering $\sigma$, whatever $\mut$ and $\Sig$.
\end{lemma}
\begin{proof}
List the coordinates in the order $\sigma$ and fix $k$. With $\mathbf v=\X-\mut_0$ and $L$
the Cholesky factor of $A$, $u^\sigma_k=(v_k-\hat\beta^\top\mathbf v_{<k})/L_{kk}$, where
$\hat\beta$ is the least-squares coefficient of coordinate $k$ on the earlier coordinates in the
reference sample (with an intercept when the mean is estimated) and $L_{kk}^2$ is that regression's
residual sum of squares; moreover $\sum_{j<k}(u^\sigma_j)^2=\mathbf v_{<k}^\top A_{<k}^{-1}\mathbf
v_{<k}$. Conditionally on the earlier coordinates of $\X$ and of the reference sample, the
prediction error $v_k-\hat\beta^\top\mathbf v_{<k}$ is normal with mean zero and variance
$\omega^2\bigl(c+\sum_{j<k}(u^\sigma_j)^2\bigr)$, where $\omega^2$ is the conditional variance of
coordinate $k$ given the earlier ones, and the residual sum of squares is an independent
$\omega^2\chi^2_{\nu_k}$. So $t^\sigma_k$ is a studentized prediction error, Student-$t$ with
$\nu_k$ degrees of freedom conditionally on those earlier coordinates. That law does not depend on
them, so $t^\sigma_k$ is independent of $t^\sigma_1,\dots,t^\sigma_{k-1}$, which are functions of
them, and the probability integral transform followed by $\Phi^{-1}$ completes the proof.
\end{proof}

The transformation is classical: it is the conditional probability integral transform of O'Reilly
and Quesenberry \cite{oreilly1973} for the multivariate normal family with mean and covariance
unknown \cite{rincon1979}, which under the group structure of this problem coincides with the
predictive probability integral transform \cite{oreillyvillegas1987}; hence our name for
$\tilde\z^\sigma$, the \emph{exact predictive scores}. What matters here is its consequence for
pooling. On $\tilde\z^\sigma$ every single-ordering statistic of \S\ref{sec:base} has under the
estimated-population null the law it has at known $F$: the Simes and Fisher $p$-values are exact
$p$-values, so their tests have size exactly $\alpha$, and $E^\sigma$ is an e-value by
Lemma~\ref{lem:evalue}, so its test at $1/\alpha$ is valid but conservative. By Lemma~\ref{lem:merge}
the pooled combiners are then valid at their nominal thresholds, and the e-value average is an
e-value, for every reference-sample size allowed by Lemma~\ref{lem:cpit}. Validity is thus
settled without calibration; what calibration adds is the removal of conservativeness. As $N\to\infty$ the exact and plug-in scores agree, and at known
$F$ there is nothing to correct. Whenever $\Sig$ is estimated below, the Rosenblatt-based tests are
computed on $\tilde\z^\sigma$. The two order-invariant baselines keep their plug-in form; for $\chi^2$
the exact null law is Hotelling's, for a monotone function of the same quadratic form, so in the
unregularized case the two coincide after calibration. With the fixed ridge $\lambda I$,
$\lambda=10^{-3}$, that the implementation adds (see below), and with $\mathbf v=\X-\mut_0$ the
centered vector, $\mathbf v^\top(\hat\Sig+\lambda I)^{-1}\mathbf v$ is not in general a monotone
function of $\mathbf v^\top\hat\Sig^{-1}\mathbf v$, so the two are close, as the ridge check below
shows, rather than identical.

\paragraph{Bootstrap calibration}
Exactness for each ordering does not make the pooled test exact. The nominal thresholds of
Lemma~\ref{lem:merge} ignore the strong dependence among orderings and are conservative, and that
dependence---the joint law of $(\tilde\z^{\sigma_1},\dots,\tilde\z^{\sigma_M})$---changes with the
correlation structure of $\Sig$, so no fixed threshold is exact.
We bring them approximately to the nominal level by a
\emph{re-estimating parametric bootstrap}. The idea is simple: since we cannot compute the exact
null distribution of our statistic under estimation noise, we \emph{manufacture} it---simulate many
worlds in which the fitted model is true, run the \emph{entire} pipeline (including re-estimating the
covariance from a fresh training sample, so the simulated worlds carry the same estimation noise as
the real one) in each world, and use the simulated statistics' quantile as the rejection threshold.
This is the standard remedy for goodness-of-fit testing with estimated parameters
\cite{stute1993,genestremillard2008}, an instance of bootstrap calibration (prepivoting) in the sense
of \cite{beran1988}; see \cite{davisonhinkley1997} for parametric and Monte-Carlo bootstrap tests
generally.

\begin{procedure}{1: Bootstrap calibration of a combiner $T$}
Given the fitted model $\hat\Sig$, target level $\alpha$, orderings $\sigma_{1:M}$, replicate count
$B$:\\
1. For $b=1,\dots,B$: draw a fresh ``training'' sample of size $\Ntr$ from $\mathcal N(\mathbf 0,\hat\Sig)$
and re-estimate $\hat\Sig^{\ast}_b$ (and the mean, when it is estimated); draw one calibration vector $\X^\ast_b\sim\mathcal N(\mathbf 0,\hat\Sig)$;
compute $T^\ast_b=T(\X^\ast_b;\hat\Sig^{\ast}_b)$ using the \emph{re-estimated} whitening---the
exact predictive scores \eqref{eq:exactz} computed from the replicate's training sample.\\
2. Rank-based Monte-Carlo $p$-value: for a $p$-value combiner (small values are evidence),
$p_{\mathrm{MC}}=\bigl(1+\#\{b:T^\ast_b\le T_{\mathrm{obs}}\}\bigr)/(B+1)$; for an e-value
combiner (large values are evidence),
$p_{\mathrm{MC}}=\bigl(1+\#\{b:T^\ast_b\ge T_{\mathrm{obs}}\}\bigr)/(B+1)$. Reject if
$p_{\mathrm{MC}}\le\alpha$.
\end{procedure}

Two remarks on Procedure~1. First, on what is claimed. The bootstrap aims at
the estimated-population null: the fitted model stands in for the unknown population, and each
replicate re-enacts the whole pipeline---estimation from $\Ntr$ draws, then whitening of an
independent vector by that estimate---so that $\{T^\ast_b\}$ mimics the sampling distribution of
$T_{\mathrm{obs}}$ when both the training sample and $\X$ come from the population. For a statistic
of a single ordering this is exact: by Lemma~\ref{lem:cpit} its null law is the same for every
$(\mut,\Sig)$, so $T_{\mathrm{obs}},T^\ast_1,\dots,T^\ast_B$ are exchangeable and $p_{\mathrm{MC}}$
is an exact $p$-value. For the pooled combiners the bootstrap reproduces the dependence among
orderings at $\hat\Sig$ rather than at $\Sig$. That approximation is asymptotic in $\Ntr$, under the
regularity conditions of the parametric-bootstrap literature
\cite{stute1993,genestremillard2008,beran1988}; we claim no finite-sample exactness for it, and
Section~\ref{sec:calib-exp} measures how close the realized size comes. The fitted null
$\X\sim\mathcal N(\mathbf 0,\hat\Sig)$ is a different hypothesis, the issued-forecast null of
Section~\ref{sec:fx}: there the plug-in scores are exactly standard normal and the calibration is
by exact simulation from the fitted model. (The exact scores use the unregularized $A$, invertible
under the conditions of Lemma~\ref{lem:cpit}. The asymptotic argument concerns simulation from
the unregularized estimate, or from one whose ridge vanishes as $\Ntr$ grows. The implementation
adds a fixed ridge $10^{-3}I$ to the covariance used by the two baselines and to the one
Procedure~1 simulates from, a numerical device outside that argument: rerunning Figure~\ref{fig:boot}
from $\Ntr=4n$ to $8n$ and at known $F$ without it changes no calibrated size by more than $0.001$
and no power by more than about $0.002$.)
The rank form of step~2 is a separate, finite-$B$ matter. When the simulated null is exact---a
fully specified $F$, where step~1 reduces to drawing $\X^\ast_b$ from $F$ and whitening it by the
same fixed factor as the data, the situation of the issued-forecast null in
Section~\ref{sec:fx}, or a single-ordering statistic on the exact predictive scores---$T_{\mathrm{obs}},T^\ast_1,\dots,T^\ast_B$ are exchangeable and
$p_{\mathrm{MC}}$ is an exact finite-sample $p$-value \cite{davisonhinkley1997}; an interpolated
empirical $\alpha$-quantile of $\{T^\ast_b\}$ is instead not exact in general at finite $B$,
and can be anti-conservative. For a
statistic that is uniform under the null, such as a $p$-value, the expected rejection probability
of the interpolated rule is $((B-1)\alpha+1)/(B+1)$, e.g.\ $0.0545$ at $B=199$ and $\alpha=0.05$
(for other continuous statistics the bias is of the same order but not distribution-free); it was
visible in earlier drafts of our own experiments and is removed by the rank rule, whose only cost
is the granularity $1/(B+1)$. Second, re-estimating
$\hat\Sig^\ast_b$ \emph{inside} each replicate is what makes the bootstrap null carry
the same estimation noise as the real pipeline; this is essential for the calibration to remove the
estimation-induced bias (Section~\ref{sec:exp}).

\subsection{The full procedure}

\begin{procedure}{2: Order-randomized Rosenblatt test (single statistic)}
Input: data vector $\X$; fitted $\hat\Sig$ and, under an estimated null, the reference sample (or
$\Ntr$, the centering point $\mut_0$, $A$ and whether the mean is known); level $\alpha$; $M$, $B$, base statistic, combiner.\\
1. Draw $\sigma_1,\dots,\sigma_M\stackrel{\text{i.i.d.}}{\sim}\Unif(S_n)$ and form
$\hat L_{\sigma_m}=\mathrm{chol}(P_{\sigma_m}\hat\Sig P_{\sigma_m}^\top)$.\\
2. Whiten: $\hat\z^{\sigma_m}=\hat L_{\sigma_m}^{-1}P_{\sigma_m}\X$, or, under an estimated
null, the exact predictive scores $\tilde\z^{\sigma_m}$ of \eqref{eq:exactz}; compute the base statistic
(\S\ref{sec:base}) for each $m$.\\
3. Combine across orderings (\S\ref{sec:combine}) to get $T_{\mathrm{obs}}$.\\
4. Calibrate the threshold by Procedure~1 and reject accordingly.
\end{procedure}

The procedure is for a single multivariate hypothesis. When a collection of such hypotheses is
screened simultaneously, the multiplicity layer depends on which null is in force: under a fully
specified null the per-hypothesis e-values $\bar E_M^{(j)}$ are valid and e-BH applies directly
(Section~\ref{sec:fx}); under an estimated null the per-hypothesis outputs are the calibrated
$p$-values of Procedure~1 and the layer is BH (Section~\ref{sec:fdr}), while e-BH on the
exact predictive e-values, valid there as well, remains available at the price of its
conservativeness. This layer is a standard add-on and is treated as a secondary
result in Section~\ref{sec:fdr}.

\section{Theory: power as a profile functional on the orbit of orderings}\label{sec:theory}

Why should pooling orderings help at all? Reordering shuffles the data through different whitenings,
but it cannot create or destroy information, so any gain must come from \emph{how} a fixed departure
is presented to the test, not from how much departure there is. This section makes that intuition
exact. We show that reordering moves the alternative around a sphere of fixed radius---the
\emph{orbit}---changing only the shape of the signal across coordinates (Lemma~\ref{lem:energy}, with
Proposition~\ref{prop:kl} as its general, non-Gaussian form: the transform's total information is
conserved under reordering); that
a pure energy test is therefore blind to ordering and cannot be helped by pooling
(Proposition~\ref{prop:chi2}); but that a shape-sensitive test sees a lucky or unlucky view under each
ordering, so averaging the views strictly improves the expected evidence by exactly the orbit's
variability---which is zero under independence (Proposition~\ref{prop:main}) and is large,
empirically, under strong dependence.
A short calculation then gives the optimal number of orderings for the max-type combiner
(Proposition~\ref{prop:optM}), and a second shows how fast the e-value average collects the
order-averaging gain as the number of orderings grows (Proposition~\ref{prop:finiteM}).

Throughout this section $\X\sim\mathcal N(\mut,\Sig)$ (the alternative), and we assume the
whitening uses the true $\Sig$ (the role of estimation is studied empirically in
Section~\ref{sec:exp}). Under ordering $\sigma$, $\z^\sigma\sim\mathcal N(\m^\sigma,I)$ with whitened
mean $\m^\sigma=L_\sigma^{-1}P_\sigma\mut$.

\begin{lemma}[Energy invariance]\label{lem:energy}
For every $\sigma\in S_n$, $\ \|\m^\sigma\|^2=\mut^\top\Sig^{-1}\mut$, and $\|\z^\sigma\|^2=\X^\top\Sig^{-1}\X$
is independent of $\sigma$.
\end{lemma}
\begin{proof}
$\Sig_\sigma=P_\sigma\Sig P_\sigma^\top=L_\sigma L_\sigma^\top$ gives
$L_\sigma^{-1}\Sig_\sigma L_\sigma^{-\top}=I$, hence $\Sig_\sigma^{-1}=L_\sigma^{-\top}L_\sigma^{-1}$.
Therefore
\[
\begin{aligned}
\|\m^\sigma\|^2
&=(P_\sigma\mut)^\top L_\sigma^{-\top}L_\sigma^{-1}(P_\sigma\mut)
=(P_\sigma\mut)^\top\Sig_\sigma^{-1}(P_\sigma\mut)\\
&=\mut^\top P_\sigma^\top(P_\sigma\Sig P_\sigma^\top)^{-1}P_\sigma\mut
=\mut^\top\Sig^{-1}\mut,
\end{aligned}
\]
using $P_\sigma^{-1}=P_\sigma^\top$. The same computation with $\X$ replacing $\mut$ gives the second
claim.
\end{proof}

Thus reordering is a \emph{norm-preserving} reparametrization: the family
$\{\m^\sigma:\sigma\in S_n\}$ lies on a sphere of fixed radius $\|\m\|=(\mut^\top\Sig^{-1}\mut)^{1/2}$,
which we call the \emph{orbit} of the alternative. What changes across $\sigma$ is the \emph{profile}
of $\m^\sigma$ (how the fixed energy is distributed across coordinates) and its sign pattern. The
quantity $\mut^\top\Sig^{-1}\mut$ is the squared \emph{Mahalanobis distance} of the alternative from
the null---the natural ``how far is the truth from the model'' scale that accounts for correlations;
we call it the signal's \emph{energy}. Lemma~\ref{lem:energy} says reordering can neither create nor
destroy energy; it only re-apportions it among the whitened coordinates. One ordering may pile the
energy onto a single coordinate (easy for a sparse-friendly test to see) while another smears the
same energy thinly over all coordinates (nearly invisible to it): same distance from the model,
very different visibility. The effect is large. With $n=50$ independent standard normal coordinates
and a mean departure of fixed total $\KL=4$, the two-sided Simes test of \S\ref{sec:base}
detects the departure $36\%$ of the time at $\alpha=0.05$ when the whole budget sits on one
coordinate and $10\%$ when it is spread evenly over all fifty (one-sided: $43\%$ and $17\%$)---equally informative alternatives, at the same distance from the null, differing only in
profile.

The energy lemma is the Gaussian instance of a fully general conservation law, which does not even
require the null to be Gaussian: however the coordinates are ordered, the transform's departure from
uniformity carries exactly the model's error, no more and no less.

\begin{proposition}[KL invariance]\label{prop:kl}
Suppose the truth is $G$ but the transform uses $F$ (both absolutely continuous, with the
conditional distribution functions of $F$ strictly increasing on their supports, so that
$R_\sigma^F$ is almost surely bijective), and let
$\U=R_\sigma^F(\X)$ for $\X\sim G$. Then $\KL(\mathcal L(\U)\,\|\,\Unif[0,1]^n)=\KL(G\,\|\,F)$, and in
particular this is invariant to $\sigma$.
\end{proposition}
\begin{proof}
$R_\sigma^F$ is a measurable bijection and $R_\sigma^F(\X)\sim\Unif[0,1]^n$ when $\X\sim F$
(Theorem~\ref{thm:rosen}). KL divergence is invariant under a common bijective transformation of its
two arguments; applying $R_\sigma^F$ to $G$ gives $\mathcal L(\U)$ and to $F$ gives $\Unif[0,1]^n$,
so $\KL(\mathcal L(\U)\,\|\,\Unif)=\KL(G\,\|\,F)$. The right side does not involve $\sigma$.
\end{proof}

(The Kullback--Leibler divergence $\KL(G\,\|\,F)$ measures, in information units, how distinguishable
the truth $G$ is from the model $F$; it is the best possible average log-likelihood ratio, so it
bounds what \emph{any} test can achieve.) Together, Lemma~\ref{lem:energy} and
Proposition~\ref{prop:kl} say that ordering changes \emph{how} a fixed departure is presented to the
test---never how much departure there is. Any gain from pooling orderings must therefore come from
presentation, and any loss from a single ordering is a presentation failure. The next two results
make both directions precise.

\begin{proposition}[Orbit-flatness of energy tests]\label{prop:chi2}
Any test whose statistic depends on $\z^\sigma$ only through $\|\z^\sigma\|^2$ has power independent
of $\sigma$: the statistic is literally the same number under every ordering, so ordering, and any
combination of orderings applied to the statistic itself, changes nothing. In particular the
$\chi^2$ test based on $\X^\top\Sig^{-1}\X$ is unaffected by ordering, and its power under a
mean-shift alternative depends on the departure only through its energy
$\mut^\top\Sig^{-1}\mut$---it is \emph{flat} across departure shapes at fixed energy.
\end{proposition}
\begin{proof}
By Lemma~\ref{lem:energy}, $\|\z^\sigma\|^2$ is the same for all $\sigma$; hence the statistic and
its distribution do not depend on $\sigma$, and under $\X\sim N(\mut,\Sig)$ the statistic is
noncentral $\chi^2_n$ with noncentrality $\mut^\top\Sig^{-1}\mut$, whatever the profile of $\mut$.
\end{proof}

(One clarification: this concerns combining the identical statistic. Pushing $M$ duplicated
copies through a \emph{penalized} nominal merger is a different operation and can only lose---
Bonferroni's $M\min_m p$ applied to $M$ equal $p$-values multiplies the $p$-value by $M$;
calibrated combination of the duplicates, on the uncapped summaries of \S\ref{sec:combine},
is exactly the original test.)

Profile-sensitive tests behave differently. The headline positive result, stated for the e-value
combiner, uses expected log-evidence (growth rate)---the canonical optimality criterion for e-values
\cite{grunwald2024}. In the betting reading, $\E_Q[\log E]$ is the long-run exponential growth rate
of repeatedly staking one's evidence capital on the bet $E$ when the truth is $Q$; maximizing it is
the evidential analogue of maximizing power, and it is the natural criterion here because log-evidence,
unlike fixed-level power, behaves cleanly under the averaging that pooling performs.

\begin{proposition}[Order-averaging dominates a single ordering]\label{prop:main}
Let $\sigma\sim\Unif(S_n)$ be drawn independently of $\X$, and for each $\sigma$ let $E^\sigma(\X)\ge0$
be an e-value, $\E_{H_0}[E^\sigma]=1$. Let $\bar E(\X)=\E_\sigma[E^\sigma(\X)\mid\X]$, and for
(ii) assume $E^\sigma(\X)>0$ $Q$-a.s.\ with $\E_Q|\log E^\sigma(\X)|<\infty$ for each $\sigma$
and $\E_Q|\log\bar E(\X)|<\infty$ (satisfied by the mixture e-value \eqref{eq:evalue}, which is
bounded below by $e^{-\max\mathcal T^2/2}/(n|\mathcal T|)>0$ and has Gaussian log-moments). Then:
\begin{enumerate}
\item[(i)] $\bar E$ is an e-value: $\E_{H_0}[\bar E]=1$.
\item[(ii)] for any alternative $Q$,
$\ \E_Q[\log\bar E(\X)]\ \ge\ \E_\sigma\,\E_Q[\log E^\sigma(\X)]$, with equality iff $E^\sigma(\X)$ is
$Q$-a.s.\ constant in $\sigma$;
\item[(iii)] if $\Sig$ is diagonal and the base e-value is exchangeable (permutation-symmetric in the
coordinatewise scores), then $E^\sigma(\X)$ is constant in $\sigma$ and (ii) holds with equality:
order combination yields no growth-rate gain.
\end{enumerate}
\end{proposition}
\begin{proof}
(i) Fubini and the e-value property: $\E_{H_0}[\bar E]=\E_\sigma\E_{H_0}[E^\sigma]=1$.
(ii) Condition on $\X$; $\log$ is concave, so Jensen gives
$\log\E_\sigma[E^\sigma(\X)\mid\X]\ge\E_\sigma[\log E^\sigma(\X)\mid\X]$, with equality iff
$E^\sigma(\X)$ is a.s.\ constant over $\sigma$ given $\X$; take $\E_Q$ and apply Fubini.
(iii) With diagonal $\Sig$, $\z^\sigma$ is a permutation of the fixed coordinatewise standardized
scores, and an exchangeable base e-value---e.g.\ the coordinate-mixture \eqref{eq:evalue}---is
invariant to the order, so $E^\sigma(\X)$ is constant in $\sigma$.
\end{proof}

The content of (ii) is elementary but worth unpacking. Conditional on the data, the per-ordering
e-values $E^\sigma(\X)$ are a collection of numbers---one view of the same observation per
ordering---and the theorem compares two policies: average the numbers first and take the log
(pooling), or take the log first and average (a single random ordering, in expectation). Because the
logarithm is concave, averaging first always wins, and the margin---the \emph{Jensen gap}---is
governed by how much the views disagree. Under independence the views are all identical
(part~(iii)) and pooling buys nothing; under strong dependence different orderings whiten the same
observation very differently, the views disagree a lot, and pooling's margin is large. This
is the mechanism we read behind the empirical pattern that the boost grows with $\rho$ in the
designs of Section~\ref{sec:exp}, but it should not be overread. The proposition is a statement
about expected log-evidence: it does not imply that the \emph{calibrated power} gain is positive
or that it grows with dependence, and the Jensen gap itself need not be monotone in $\rho$---for
$n=2$ and $\rho\to1$ the two Cholesky whitenings agree up to the signs of the coordinates, so a
sign-invariant base statistic such as \eqref{eq:evalue} has a Jensen gap that vanishes at fixed
energy. The empirical sections verify that the power boost indeed materializes, quantify it, and
map where it is large or negligible.

\begin{remark}[When expected log-evidence controls power]\label{rem:bridge}
Write $\mu=\E_Q[\log E]$ and $s^2=\operatorname{Var}_Q(\log E)$, and reject when $E\ge1/\alpha$,
equivalently $\log E\ge t:=\log(1/\alpha)$---the threshold Markov's inequality licenses for any
e-value. Two bridges lead from $\mu$ to power.

\emph{(i) Repeated observations.} Suppose $T$ i.i.d.\ observations from $Q$ are tested and
their e-values multiplied---the natural use of evidence, valid with no correction---and that
$\mu>0$ with $\E_Q|\log E|<\infty$. Then $\log\prod_{i\le T}E_i$ has mean $T\mu$, and by the law of
large numbers the number of observations needed to reach any fixed target power $1-\varepsilon$,
$T_{\varepsilon}(\alpha):=\min\{T:\Prob_Q(\sum_{i\le T}\log E_i\ge t)\ge1-\varepsilon\}$,
satisfies $T_{\varepsilon}(\alpha)\sim t/\mu$ as $\alpha\downarrow0$ (i.e.\ $t\to\infty$).
Asymptotically, then, the required sample size is \emph{inversely proportional} to expected
log-evidence: a procedure with the larger $\mu$ needs proportionally less data for the same power.
The approximation is informative well before the limit. In the design of
Figure~\ref{fig:power} with the dependence swept over $\rho\in\{0.2,0.5,0.8\}$, pooling raises $\mu$
by factors $1.027$, $1.108$ and $1.114$, and an interpolated proxy for the reduction in the number of observations
needed for power $0.80$ (linear interpolation between consecutive integers $T$; at the integers
themselves both procedures first reach $0.80$ at $T=3$) gives factors $1.035$, $1.117$ and $1.125$---agreement to within one percent
at $T$ between two and three.

\emph{(ii) A single observation, strong signal.} At $T=1$, Cantelli's one-sided Chebyshev
inequality \cite{cantelli1928,blm2013} gives
\[
\Prob_Q(\log E\ge t)\ \ge\ \frac{(\mu-t)_+^2}{s^2+(\mu-t)^2},
\]
increasing in $\mu$ and decreasing in $s^2$: raising expected log-evidence without inflating its
dispersion raises a lower bound on power. The bound is informative only when $\mu>t$---when the
expected evidence already exceeds the level one is testing at, $2.996$ nats at $\alpha=0.05$. That
condition fails through most of our designs, which is consistent with our fixed-level gains being a
fraction of the evidence gains rather than proportional to them; it does not by itself explain the
shortfall, since a larger mean and a smaller variance of $\log E$ do not in general imply a larger
upper-tail probability. Where it holds it orders the two
procedures correctly, and pooling is then measured to reduce $s^2$ as well as raise $\mu$, so both
terms move the same way: at $\rho=0.8$ and $\mathrm{ncp}=20$ the bound rises from $0.12$ to $0.30$
under pooling while true power rises from $0.59$ to $0.72$, and at $\mathrm{ncp}=30$ from $0.64$ to
$0.72$ against true power $0.90$ to $0.95$.
\end{remark}

\begin{proposition}[A finite number of orderings]\label{prop:finiteM}
In the setting of Proposition~\ref{prop:main}, suppose $E^\sigma(\X)>0$ for every $\sigma$, let
$\sigma_1,\dots,\sigma_M$ be drawn independently from $\Unif(S_n)$, independently of $\X$, and let
$\bar E_M=M^{-1}\sum_{m=1}^M E^{\sigma_m}$. Write
$V(\X)=\mathrm{Var}_\sigma[E^\sigma(\X)\mid\X]/\bar E(\X)^2$ for the spread of the per-ordering
e-values, and $a(\X)\le1\le b(\X)$ for the smallest and largest of $E^\sigma(\X)/\bar E(\X)$ over
$\sigma\in S_n$. Then:
\begin{enumerate}
\item[(i)] $\bar E_M$ is an e-value, and $\E_Q[\log\bar E_M(\X)]$ is nondecreasing in $M$, from
$\E_\sigma\E_Q[\log E^\sigma(\X)]$ at $M=1$ to $\E_Q[\log\bar E(\X)]$ as $M\to\infty$;
\item[(ii)] the expected log-evidence still missing at $M$,
$D_M(\X)=\log\bar E(\X)-\E[\log\bar E_M(\X)\mid\X]$, satisfies
\[
\frac{V(\X)}{2\,b(\X)^2M}\;\le\;D_M(\X)\;\le\;\frac{V(\X)}{2\,a(\X)^2M}
\]
and $M\,D_M(\X)\to\tfrac12V(\X)$ as $M\to\infty$; if $\E_Q[V(\X)/a(\X)^2]<\infty$, then also
$M\,\E_Q[D_M(\X)]\to\tfrac12\E_Q[V(\X)]$.
\end{enumerate}
\end{proposition}
\begin{proof}
(i) $\E_{H_0}[\bar E_M]=1$ by Fubini, as in Proposition~\ref{prop:main}(i). Given $\X$, the
average of $M+1$ draws is the mean of its $M+1$ leave-one-out averages, each distributed as
$\bar E_M$, so concavity of $\log$ gives $\E[\log\bar E_{M+1}\mid\X]\ge\E[\log\bar E_M\mid\X]$;
take $\E_Q$. The limit follows from (ii) by monotone convergence, since $0\le D_M\le D_1$, $D_1$ is
$Q$-integrable under the assumptions of Proposition~\ref{prop:main}(ii), and $D_M\to0$.
(ii) Put $x=\bar E_M/\bar E$. Given $\X$, $a\le x\le b$, $\E[x\mid\X]=1$ and
$\mathrm{Var}(x\mid\X)=V/M$. For $x\in[a,b]$ with $0<a\le1\le b$,
$(x-1)-(x-1)^2/(2a^2)\le\log x\le(x-1)-(x-1)^2/(2b^2)$: each difference vanishes at $x=1$,
decreases to its left and increases to its right. Taking conditional expectations gives the
bounds. For the limit, $-\log x=-(x-1)+(x-1)^2/(2\xi^2)$ with $\xi$ between $1$ and $x$; $\xi\to1$
in probability, $1/\xi^2\le1/a^2$, and $M(x-1)^2$ is uniformly integrable: $x$ averages $M$
conditionally independent copies of $Y=E^\sigma/\bar E$, so
$\E[M^2(x-1)^4\mid\X]=\mu_4/M+3(1-1/M)V^2$ with $\mu_4=\E_\sigma[(Y-1)^4\mid\X]$ finite; hence
$M D_M\to V/2$. The averaged limit follows by dominated convergence, since $M D_M\le V/(2a^2)$. Its
condition holds for the mixture e-value \eqref{eq:evalue} under any Gaussian alternative: each
$E^\sigma$ lies between $c=|\mathcal T|^{-1}\sum_{\tau\in\mathcal T}e^{-\tau^2/2}$ and
$e^{\tau_{\max}\|\z^\sigma\|_\infty}$, and $V\le b$, so
$V/a^2\le(b/a)^2\le c^{-2}\exp(2\tau_{\max}\max_\sigma\|\z^\sigma\|_\infty)$, which has finite
expectation because each $\z^\sigma$ is a linear function of the Gaussian $\X$.
\end{proof}

Both parts have precedents for logarithms of Monte Carlo averages: the monotonicity in (i) is
the argument of \cite{burda2016} for importance-weighted bounds, and the rate in (ii) is the
relative-variance expansion of \cite{maddison2017}. Here they say that the order-averaging gain of
Proposition~\ref{prop:main} arrives at rate $1/M$: for large $M$ the log-evidence still missing is
$V/(2M)+o(1/M)$ nats, and doubling $M$ halves it. This suggests a starting value for $M$. Write
$\bar V_Q=\E_Q[V(\X)]$ for the average spread under a distribution $Q$ of the data; it can be
estimated before the data are seen by simulating from $Q$, and $M\approx\bar V_Q/(2\varepsilon)$
then aims at an average loss of about $\varepsilon$ nats under $Q$. Two cautions apply. The rule is
asymptotic, and $V$ alone gives no finite-$M$ guarantee: the bound in (ii) carries the factor
$1/a^2$, and when a few orderings give e-values far below the average, the loss at moderate $M$ can
be several times $V/(2M)$, so the proposed $M$ should be checked by simulation at that $M$. And the
recommendation targets the distribution simulated: the spread is much smaller under the null than
under alternatives (\S\ref{sec:which}), so a count tuned under the null need not retain the
evidence under an alternative, which is what averaging is for. The count may also depend on the
data, provided the orderings that enter $\bar E_M$ are drawn uniformly afresh, independently of
$\X$ and of whatever was used to choose $M$; then $\E[\bar E_M\mid\X,M]=\bar E(\X)$ and $\bar E_M$ is
still an e-value. Reusing the pilot orderings that determined $M$, or stopping the average when it
looks large, is not covered, and under bootstrap calibration the whole selection rule, pilot stage
included, is repeated in every replicate.

A practitioner also wants to know whether another random set of orderings would change the
decision.
\begin{corollary}[Stability of the decision]\label{cor:stable}
In the setting of Proposition~\ref{prop:finiteM}, fix a threshold $c>0$ and $\delta\in(0,1)$. Then
\[
\Prob\big(\mathbf 1\{\bar E_M\ge c\}\ne\mathbf 1\{\bar E\ge c\}\big)\ \le\
\Prob_Q\big(|1-c/\bar E(\X)|\le\delta\big)+\frac{\E_Q[V(\X)]}{M\delta^2},
\]
the probability on the left being over $\X\sim Q$ and the orderings, and two independent sets of
$M$ orderings reach different decisions with at most twice this probability.
\end{corollary}
\begin{proof}
Put $x=\bar E_M/\bar E$. If $c/\bar E<1-\delta$, then $\bar E>c$ and the decisions differ only if
$x<c/\bar E$, so $|x-1|>\delta$; if $c/\bar E>1+\delta$, then $\bar E<c$ and they differ only if
$x\ge c/\bar E$, so again $|x-1|>\delta$. Given $\X$, Chebyshev's inequality and
$\mathrm{Var}(x\mid\X)=V/M$ bound the probability of that event by $V/(M\delta^2)$; take $\E_Q$. Two
sets of orderings can disagree only if one of them disagrees with $\mathbf 1\{\bar E\ge c\}$.
\end{proof}

The first term counts observations within a relative distance $\delta$ of the threshold; the second is the Monte Carlo error of the ordering average. The first term
belongs to an upper bound and is not an error floor: for an observation with $\bar E(\X)\ne c$ the
decision converges as $M\to\infty$, so when $Q\{\bar E(\X)=c\}=0$ the disagreement tends to zero.
What the bound does say is that observations near the threshold settle slowly, since a small first
term needs a small $\delta$ and the second term then needs $M$ of order $1/\delta^2$. With $c=1/\alpha$
the corollary covers the e-value test at known $F$. With calibrated thresholds the critical value
also moves with the set of orderings, which the corollary does not cover, so \S\ref{sec:which}
measures the disagreement directly.

\begin{proposition}[Optimal number of orderings for Bonferroni]\label{prop:optM}
Fix the data at an alternative configuration and suppose that, over the random draw of $\sigma$, the
per-ordering Simes $p$-value equals a small value $p_g$, $0<p_g\le\alpha$, with probability
$q\in(0,1)$ (a ``concentrating'' ordering) and is $>\alpha$ otherwise. Draw $M$ i.i.d.\ orderings. Then the
Bonferroni-over-orders test (reject iff $M\min_m p^{\sigma_m}\le\alpha$) has power
\[
\beta(M)=\begin{cases}1-(1-q)^M, & M\le \alpha/p_g,\\[2pt] 0, & M>\alpha/p_g,\end{cases}
\]
so $\beta$ is increasing on the feasible region and the optimal count for this data configuration is
$M^\ast=\lfloor\alpha/p_g\rfloor$; the optimal power $1-(1-q)^{M^\ast}$ reaches
$1-\varepsilon$ once $M^\ast\ge\lceil\log\varepsilon/\log(1-q)\rceil$, i.e.\ roughly
$p_g\le\alpha\log(1-q)/\log\varepsilon$
(at $M^\ast\approx1/q$ it is only about $1-e^{-1}\approx0.63$ for small $q$).
\end{proposition}
\begin{proof}
A ``bad'' ordering has $p>\alpha\ge\alpha/M$ for $M\ge1$, so it never triggers; only ``good''
orderings (value $p_g$) can, and they do iff $p_g\le\alpha/M$, i.e.\ $M\le\alpha/p_g$. Conditional on
that, rejection occurs iff at least one of the $M$ i.i.d.\ orderings is good, with probability
$1-(1-q)^M$, which is increasing in $M$. Hence within the feasible region the largest admissible $M$
is optimal, $M^\ast=\lfloor\alpha/p_g\rfloor$.
\end{proof}

\begin{remark}\label{rem:dome}
Proposition~\ref{prop:optM} is a statement about the \emph{nominal} $\alpha/M$ threshold, and its
hard cliff is an artifact of the stylized two-point model. With real alternatives the per-ordering
$p$-value is not
two-point: for a given $\X$ it takes one value per ordering, up to $n!$ distinct values, often many
of them near zero, so as $M$ grows the best of $M$ draws can shrink almost as
fast as the threshold tightens. In the
designs we examined the cliff accordingly softens into a shallow dome, an empirical pattern rather
than a consequence of the model: power rises steeply, peaks at a moderate $M^\ast$, then declines slowly. The decline is
eventual for every fixed signal: the smallest per-ordering $p$-value over the finite set of all
orderings, $p_{\min}(\X)=\min_{\sigma\in S_n}p^\sigma(\X)$, is positive almost surely, so
$M\min_{m\le M}p^{\sigma_m}\ge Mp_{\min}\to\infty$ and the nominal Bonferroni test eventually
loses all power as $M\to\infty$. Within the range of $M$ we tested, the decline was visible only for
weak signals (in a hub design with a weak one-coordinate shift, power peaks at $M=2$ with $0.933$ and
decays to $0.877$ by $M=512$; at stronger signals the curve saturates at $1$ and no decline is
visible up to $M=512$). The averaging combiner shows no dome in the same sweep---its nominal-threshold power
was nondecreasing in $M$ in every configuration we ran; we note that this is an empirical
statement about these designs, not a monotonicity theorem for fixed-level power, which
Proposition~\ref{prop:main} does not supply. None of this survives calibration, which is what
we recommend in practice: with
bootstrap-calibrated thresholds the $\times M$ penalty is absorbed into the critical value and
calibrated power rises over the first few orderings and then levels off for every
combiner in the sweep of \S\ref{sec:which}. The
practical reading is therefore simply that a modest $M$ suffices, and that in every calibrated
configuration we ran, ``too many'' orderings did not hurt. Proposition~\ref{prop:finiteM} and
\S\ref{sec:which} make ``modest'' precise.
\end{remark}

\begin{remark}[Two orthogonal levers; calibration unlocks the orbit lever]
Improving a profile-sensitive test has two independent axes: (a) choosing a profile- and sign-robust
base statistic (the two-sided bet)---a property of the \emph{test}; and (b) combining across
orderings to escape an unfavorable profile---a property of the \emph{aggregation}. With nominal
thresholds, the combiners carry multiplicity/conservativeness penalties (Bonferroni's $\times M$, the
merge's $\times2$, the e-value's Markov slack) that can erase the orbit advantage and make a single
ordering look as good or better. Bootstrap calibration (\S\ref{sec:calib}) removes these
penalties---and under
estimation, where the exact predictive scores already make every nominal test valid, it reduces
their conservativeness, bringing the pooled combiners approximately to the nominal level---so the richer combined
statistic realizes its advantage. This is the mechanism behind the headline result.
\end{remark}

\section{The base statistic must be two-sided}\label{sec:misspec}

The Rosenblatt scores are standardized \emph{conditional} quantities, and the sign with which a
departure surfaces in them depends on the ordering: in Example~\ref{ex:bivar} the same observation
produced $U=1.8\times10^{-5}$ under one ordering and $U=0.999998$ under the other---the identical
surprise, delivered once to the lower and once to the upper tail. So under reordering, wrong-tail
signal is the norm, not an edge case, and the classical one-sided Simes combination---which fixes
one tail in advance; in Figure~\ref{fig:blindspot} it is the upper tail, $p_k=1-\Phi(z_k)$, so
only large \emph{positive} scores register as evidence---is unusable when the sign of the
departure is not known in advance. Its failure is worse than
lost power: a
wrong-sign departure sends the affected $p$-value toward $1$, where the statistic quietly
discards it while the critical constants $\alpha i/n$ remain scaled for all $n$ coordinates, so the
rejection probability falls strictly \emph{below} the level and stays there however extreme the data
become (as the shift grows it converges to $\alpha(n-1)/n$, the level carried by the $n-1$ null
coordinates---an exact consequence of the Simes null identity, and visible as the plateau at
$0.047\approx0.05\times19/20$ in Figure~\ref{fig:blindspot}A). Nor can the failure be repaired by
simply testing the other tail: whitened departures generically carry \emph{both} signs at once (the
hub example of \S\ref{sec:exp} whitens to one positive and nineteen negative components), and
Figure~\ref{fig:blindspot}B makes the point cleanly by shifting half the coordinates by $+c/\sqrt n$
and half by $-c/\sqrt n$---either one-sided orientation then sees at most half the signal (by
symmetry both give the same power curve), and the two-sided test dominates uniformly in $c$. We
therefore fold every coordinate two-sided, $p_k=2\,\Phi(-|z_k|)=2\min(U_k,1-U_k)$: the fold acts
coordinatewise, so the folded values remain exactly i.i.d.\ uniform under $H_0$, and evidence
registers whichever tails the whitening chooses.

\begin{figure}[t]\centering
\includegraphics[width=\linewidth]{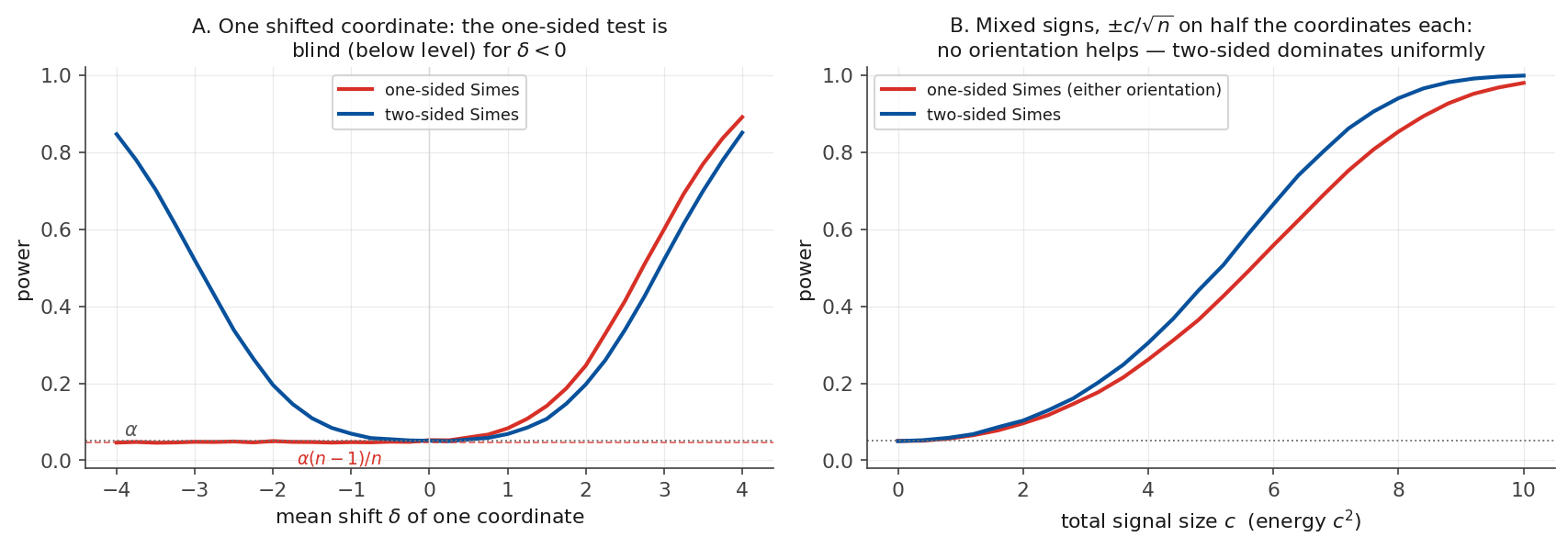}
\caption{One- vs.\ two-sided Simes ($n=20$; independent null, where all orderings coincide). (A)
Signed single-coordinate shift: the one-sided test is blind to $\delta<0$ and sits strictly below
its level (dashed: the limit $\alpha(n-1)/n$); the two-sided fold restores symmetric detection at
negligible cost for $\delta>0$. (B) Mixed-sign shift, $\pm c/\sqrt n$ on half the coordinates each:
no one-sided orientation can see more than half the signal, and the two-sided test dominates
uniformly---the deficiency is structural, not a matter of choosing the right tail.}
\label{fig:blindspot}
\end{figure}

\section{Calibration and the payoff: Gaussian experiments}\label{sec:exp}

With the orbit established (Section~\ref{sec:theory}) and the base test fixed
(Section~\ref{sec:misspec}), only the third ingredient remains. We use the experiments to do two
things at once: show that under an estimated null the naive plug-in pooled test is
miscalibrated, that the exact predictive scores make every nominal test valid and every single-ordering $p$-value exact, and that the bootstrap repairs the pooled tests (\S\ref{sec:calib-exp})---and then, on the level
playing field calibration creates, measure the payoff the theory motivates: a power gain over the
expected single random ordering, one that grows with dependence in these designs\ and persists
under estimation, together with a power advantage over the order-invariant baselines at known $F$
(\S\ref{sec:headline}) and, under estimation, at every training size
in Figures~\ref{fig:boot} and~\ref{fig:estF} and, in Table~\ref{tab:estgrid}, for every departure
shape but a single shifted coordinate, where the symmetric-root test keeps an edge of at most
$0.006$ up to $\Ntr=8n$. Practical guidance on the choice of combiner is deferred to \S\ref{sec:which}, since
it turns on evidence from both this section and the screen of \S\ref{sec:fdr-main}.

\paragraph{Protocol}
Unless stated otherwise: the null is $\mathcal N(\mathbf 0,\Sig)$ with $\Sig$ equicorrelated,
$\Sig_{ij}=\rho$ for $i\neq j$ and $1$ on the diagonal; the alternative shifts a single coordinate
(``sparse'') or all coordinates equally (``dense''), with the squared Mahalanobis distance of the
alternative from the null---the noncentrality parameter $\mathrm{ncp}$ of
\S\ref{sec:baselines}---fixed at
$\mathrm{ncp}=12$ in \emph{every} experiment so that detectability is
comparable across configurations (holding energy fixed is what makes ``sparse vs.\ dense'' or
``$\rho=0.2$ vs.\ $\rho=0.8$'' comparisons meaningful: every alternative is equally far from the
null, only its shape changes); level $\alpha=0.05$; bet grid $\mathcal T=\{1,2,3\}$; $M=12$ random
orderings; calibration by Procedure~1 with $B$ replicates. We estimate $\Sig$ by the sample
covariance from $\Ntr$ null vectors (``heavy estimation'' uses $\Ntr=4n$; ``known $F$'' uses the true
$\Sig$), with the mean known to be zero; under estimation the Rosenblatt-based tests use the
exact predictive scores of \S\ref{sec:calib}, and the scripts also record the plug-in scores on the
same draws, which we quote where the comparison is instructive. Each reported number averages over $R$ Monte-Carlo realizations of the fitted model, each
with a batch of fresh test vectors; all seeds are fixed integers. Power is compared \emph{at matched
calibrated level}---exactly at known $F$, and approximately under estimation, where calibrated
sizes range from $0.042$ to $0.060$ (\S\ref{sec:calib-exp})---so differences are genuine power differences rather than calibration artifacts, up
to that approximation.
For the headline runs, $B\in\{199,299\}$, $R\in\{150,300\}$; full settings are in the scripts
(\ref{app:repro}).

Two design points deserve emphasis. \emph{First}, the single-ordering baseline is drawn \emph{afresh
in every realization}, so its average estimates the expected power of an arbitrarily chosen
ordering---the quantity the theory speaks to. This matters: the power of one fixed ordering depends
strongly on which coordinates precede the departure, so freezing a single unlucky permutation
across realizations can inflate the apparent gain from pooling several-fold. Because the
single-ordering baseline is a Simes $p$-value while the e-value average pools mixture e-values,
we also carry the calibrated mixture e-value of one random ordering (``single-ordering e-value''),
so that the e-value path can be compared like with like. \emph{Second}, the two order-invariant references of
\S\ref{sec:baselines}---the $\chi^2$ energy test and the symmetric-root test---are carried through
every table, both bootstrap-calibrated by the same Procedure~1 with $\hat\Sig$ in place of $\Sig$. The
alternatives we study are \emph{location} (mean-shift) departures; a systematic evaluation against
scale, dependence, and shape alternatives is left to future work.

\subsection{Calibration restores level---and is necessary}\label{sec:calib-exp}

Under an estimated covariance the na\"ive plug-in test is badly miscalibrated, in \emph{opposite}
directions for different combiners, and each direction has a clear cause.
The min-type statistics are \emph{anti}-conservative---at $\Ntr=4n$ the single-ordering Simes
test on plug-in scores realizes size $0.20$, Fisher $0.28$, and the two order-invariant
references, which keep their plug-in form, $0.28$ ($\chi^2$) and $0.19$ (symmetric root;
Figure~\ref{fig:boot}A)---because whitening
with a noisy $\hat\Sig$ leaves heavier-than-normal tails in the scores and a min-type statistic latches onto
whichever coordinate the estimation error happens to make extreme. The exact predictive
scores remove this at its source: on them the single-ordering Simes and Fisher tests have size
between $0.046$ and $0.054$ at every training size (Figure~\ref{fig:boot}A), as
Lemma~\ref{lem:cpit} guarantees. On those scores the pooled combiners err the
other way. The e-value average spends $0.001$ of its level at every training size: its Markov threshold
$1/\alpha$ is loose, so most of the level goes unused. Bonferroni-over-orders ($0.014$ to
$0.020$) and the $p$-merge ($0.008$ to $0.012$) are
conservative for a related reason, their union bound and factor $2$ ignoring that the $M$ orderings
are strongly dependent recomputations of one observation. The re-estimating bootstrap (Procedure~1)
brings \emph{every} method to size $\approx\alpha$ at every training size (Figure~\ref{fig:boot}B,
drawn on the same scale as panel A). Calibration is therefore  the
enabler of a fair power comparison, and on plug-in scores a validity requirement as well: without it, the power ranking of the methods would mostly
reflect how much level each one wastes. Size control after calibration is not confined to this
configuration: across all $64$ configuration$\times$method cells of the wider robustness sweep the
calibrated size averages $0.050$ (range $0.042$--$0.060$).

Figure~\ref{fig:boot} shows the whole pipeline---miscalibration, repair, and payoff---in a single
setting: $n=20$, $\rho=0.5$, two same-sign shifted coordinates (one step off the one-coordinate
knife edge where the symmetric-root test is at its best, \S\ref{sec:headline}). At calibrated level the
pooled $p$-merge outperforms both order-invariant references of \S\ref{sec:baselines}, the single
random ordering and the single-ordering Fisher test, at every training size and at known $F$.
Under estimation part of that edge may be size: the $p$-merge runs at $0.058$ and $0.060$ at
$\Ntr=4n$ and $5n$, against $0.049$ and $0.053$ for the symmetric-root test. The e-value average,
whose calibrated size stays between $0.046$ and $0.053$, also beats both references at every
training size.

\begin{figure}[t]\centering
\includegraphics[width=\linewidth]{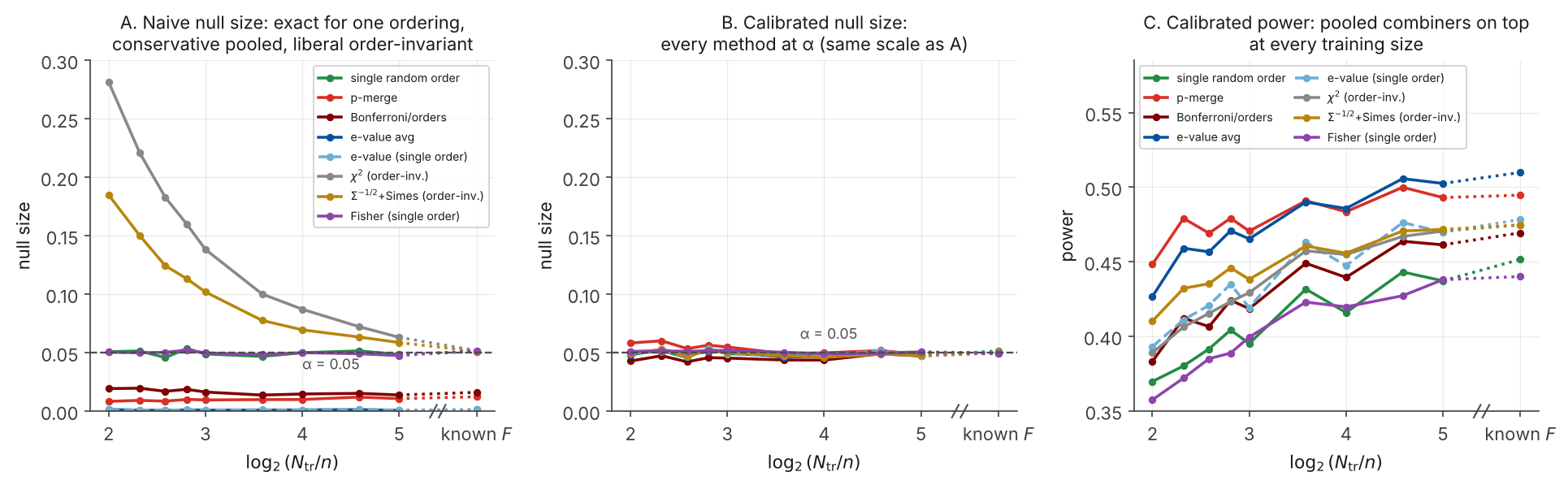}
\caption{One setting, whole pipeline ($n=20$, $\rho=0.5$, two same-sign shifted coordinates,
$\mathrm{ncp}=12$, $M=12$). Training-sample size is on a log scale, from $\Ntr=4n$ to $32n$, with known $F$ after the break. Under estimation the Rosenblatt-based tests use the exact
predictive scores. (A) Na\"ive thresholds under an
estimated null: the single-ordering tests are exact, the pooled combiners conservative (down to
$0.001$ for the e-value average), and the order-invariant references, on plug-in scores,
anti-conservative ($0.28$ for $\chi^2$ at $\Ntr=4n$). (B) After
bootstrap calibration every method sits at $\alpha$ at every training size (same $y$-scale as A).
(C) On that level playing field the pooled p-merge outperforms every order-invariant reference at
every training size and at known $F$ (at $\Ntr=4n$, $8n$ and known $F$: $0.448$, $0.471$, $0.495$ against the best reference, the
symmetric-root test, at $0.410$, $0.438$, $0.475$, which $\chi^2$ ties at known $F$); under
estimation part of the $p$-merge's edge may come from its calibrated size, up to $0.060$, and the
e-value average, closer to $\alpha$, also leads both references at every training size. The dashed line is the calibrated
\emph{single-ordering} mixture e-value; its distance to the e-value average (at the same three points: $+0.034$, $+0.046$,
$+0.032$) is the gain from aggregation alone, the base statistic held fixed.}
\label{fig:boot}
\end{figure}

\subsection{Power across departure shapes, and a gain that grows with dependence}\label{sec:headline}

The practical case for pooling is visible in Figure~\ref{fig:power}, which compares the tests
across departure shapes at equal energy.
Against a \emph{sparse} departure the tests not tuned to concentration trail badly (Fisher $0.578$, $\chi^2$
$0.629$) while the symmetric-root test is the strongest single competitor ($0.723$). Against a
\emph{dense} departure of the same energy the ranking inverts violently: the symmetric-root test
collapses to $0.376$---its whitening spreads a dense signal thinly over all coordinates, exactly
where a Simes-type statistic is weakest, whereas the Cholesky whitenings concentrate it---and
$\chi^2$ moves to the front of the invariant references, with the single-ordering Fisher test
close behind; $\chi^2$ moves up not by gaining power, which is flat across shapes, but because the
shape-sensitive tests lose it. Neither order-invariant reference reaches the top two in either
regime, whereas the pooled combiners hold
the top two places in each: $p$-merge $0.730$ and the e-value average $0.728$ against the sparse
departure, Bonferroni-over-orders $0.691$ and the e-value average $0.662$ against the dense one.
Which \emph{pair} leads changes with the shape---$p$-merge drops to $0.577$ on the dense departure
and Bonferroni to $0.713$ on the sparse---and the e-value average is the only test in the top two of
both, which is the finding \S\ref{sec:which} takes up. In both regimes, then, the e-value average is at
least as powerful as either order-invariant reference, so an analyst who must fix a test before
knowing which regime applies is better served by it than by either.

The gain over the \emph{expected} single random ordering---the baseline
Proposition~\ref{prop:main} speaks to, in expected-log-evidence terms---grows with dependence
in this sweep, an empirical regularity the proposition motivates but does not assert: at known
$F$ it rises from $+0.017$ at $\rho=0.2$ to $+0.042$ and $+0.063$ at
$\rho=0.5$ and $0.8$ for the e-value average, and $+0.018$, $+0.041$, $+0.059$ for the p-merge.
The gains are modest---strong dependence enlarges the orbit, but most orderings of an
equicorrelated $\Sig$ are reasonably favorable---and statistically unambiguous, with paired
$t\ge5$ in every known-$F$ cell. A cleaner decomposition compares the e-value average with
the calibrated mixture e-value of a \emph{single} ordering (Figures~\ref{fig:boot} and
\ref{fig:power}), which holds the base statistic fixed and isolates aggregation: at known $F$
that gap is $+0.012$, $+0.033$ and $+0.049$ at $\rho=0.2$, $0.5$ and $0.8$, $+0.025$ on the sparse
and $+0.095$ on the dense departure of Figure~\ref{fig:power}, and $+0.021$ to $+0.040$ across
$n=5$, $10$, $20$ (paired $t\ge4.7$ in every known-$F$ cell, $t\ge4.6$ under estimation).

\begin{figure}[t]\centering
\includegraphics[width=0.78\linewidth]{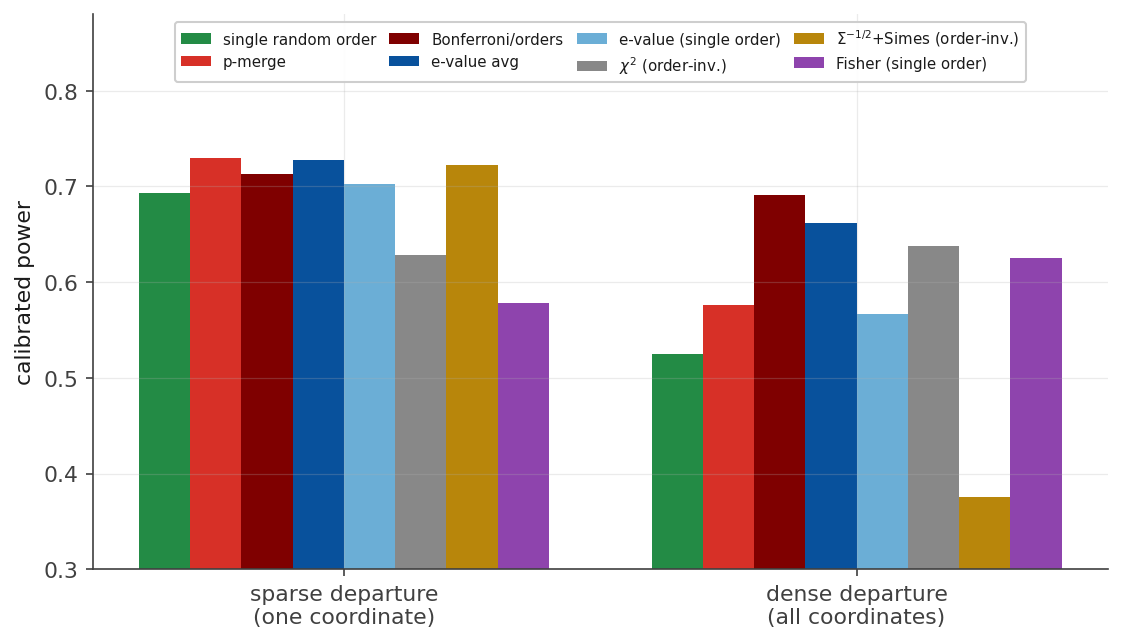}
\caption{Calibrated power against departure shape at equal energy ($n=10$, $\rho=0.5$,
$\mathrm{ncp}=12$, $M=12$, known $F$); all tests are bootstrap-calibrated to size $\alpha$, so the
comparison is at matched level. The order-invariant references swap places between the two regimes:
the symmetric-root test is strong on the sparse departure and collapses on the dense one, Fisher
does the reverse, and $\chi^2$---flat across shapes by Proposition~\ref{prop:chi2}---moves up the
ranking on the dense departure only because the others move. The e-value average is in the top two of both. The light-blue
bar is the calibrated single-ordering mixture e-value, the base statistic of the e-value average
without aggregation.}
\label{fig:power}
\end{figure}

\section{Screening many hypotheses}\label{sec:fdr}

\subsection{Calibrated pooling in the multiplicity layer}\label{sec:fdr-main}

When a collection of multivariate hypotheses $H_0^{(j)}:\X^{(j)}\sim F$, $j=1,\dots,N$, is screened
simultaneously, the single-hypothesis test of Section~\ref{sec:method} plugs into a standard
multiplicity layer---BH on per-hypothesis $p$-values, or e-BH on per-hypothesis e-values
$\bar E_M^{(j)}$. It introduces no new methodology, but it is the setting in which the advantage over the
order-invariant baselines seen in \S\ref{sec:headline} matters most, because a screen rarely faces one kind of
departure: we let the non-null hypotheses fail heterogeneously, half on a single coordinate and half
on all of them, at common energy $\mathrm{ncp}=20$ ($n=20$, $\rho=0.5$, $200$ hypotheses,
$\pi_0=0.95$, $q=0.10$, $R=200$ replications). One design point matters: the bootstrap replicate
count should satisfy $B>N/q$, since a bootstrap $p$-value cannot fall below $1/(B+1)$ while BH's
\emph{smallest} threshold is $q/N$---the condition guarantees the whole BH ladder is reachable
(it is not necessary for making rejections: the $k$-th threshold is $qk/N$, so with $B=999$
ten non-null $p$-values of $0.001$ would all be rejected at the tenth rung, $0.005$; what a smaller
$B$ forfeits is the bottom of the ladder);
we use $B=2500$.

Two findings, both in Figure~\ref{fig:estF}. First, validity again hinges on calibration or on the exact scores, and
again fails in both directions without them: with the na\"ive plug-in, BH realizes FDR $0.80$ for $\chi^2$,
$0.66$ for a single ordering and $0.61$ for the symmetric-root test against a target of $0.10$ at
$\Ntr=4n$, while e-BH on the conservative e-value average spends nothing at all (realized FDR
$0.000$ at known $F$). On the exact predictive scores the Rosenblatt-based tests keep
the FDR near or below target before any calibration: BH realizes $0.09$ to $0.11$ on a single ordering and at most
$0.05$ on Bonferroni-over-orders at every training size, and e-BH on the exact e-value average is
valid under arbitrary dependence, though with almost no power (at most $0.09$). Calibrating each per-hypothesis statistic by Procedure~1 and applying BH to
the resulting bootstrap $p$-values brings every method to target---as an empirical matter:
the $N$ test vectors are independent, but their calibrated $p$-values share the estimated
$\hat\Sig$ and the bootstrap pool, so BH's positive-dependence condition is not established here
and the realized-FDR figures are evidence of control under that shared calibration rather than a
theorem---with one exception worth naming:
$p$-merge realizes $0.152$ at $\Ntr=4n$, five standard errors above $q\pi_0=0.095$, the
bound BH guarantees under independence or positive dependence. Two things may contribute: the
re-estimating bootstrap is an asymptotic device---on the exact scores it is exact for each
ordering but reproduces the dependence among orderings only through $\hat\Sig$---and at four observations per
dimension it is not yet exact, and the shared calibration makes the $p$-values dependent in a way
that guarantee does not cover; a sparse screen, where few rejections make the realized proportion
sensitive to the tail of the calibrating distribution, is where either would show. From $5n$ on the excess is smaller
($0.102$ to $0.126$), within two standard errors except at $6n$ and $7n$, and at
known $F$ there is none. At known $F$ the guarantee does apply: with the orderings fixed within a
replication and exact null draws independent of the test vectors, the calibrated $p$-values are
conformal $p$-values sharing one calibration sample, which are PRDS \cite{bates2023} (the
calibrated statistics are continuous, as that result requires, because the uncapped summaries are
calibrated), so BH
controls the FDR at $q\pi_0$. Second, on that level playing field the pooled combiners dominate: at known $F$,
Bonferroni-over-orders detects $0.642$ of the non-nulls and the e-value average $0.619$, against
$0.474$ for a single ordering, $0.353$ for the symmetric-root test and $0.349$ for $\chi^2$; the
same ordering holds under estimation ($0.579$ and $0.576$ against $0.397$, $0.320$ and $0.240$ at
$\Ntr=8n$; both lead the other three at every training size). Bonferroni-over-orders leads here although it trails in the single-hypothesis
comparison of Figure~\ref{fig:boot}; \S\ref{sec:which} takes that up next, the short version being
that BH interrogates each statistic at $qk/N$ rather than at $\alpha$. The reason the pooled combiners
lead as a group is the pattern of
\S\ref{sec:headline}: the symmetric-root test recovers the sparse half of the non-nulls and misses
the dense half, $\chi^2$, whose power does not depend on the shape, is mediocre on both, and only the pooled tests are competitive on each. A
final advantage of e-BH, under a fully specified null or on the exact predictive scores
under an estimated one, where the $\bar E_M^{(j)}$ are e-values,
is that it controls FDR under \emph{arbitrary} dependence among the
hypotheses---relevant when the $\X^{(j)}$ share data or structure---whereas BH on $p$-values
requires independence or positive dependence.

\begin{figure}[t]\centering
\includegraphics[width=\linewidth]{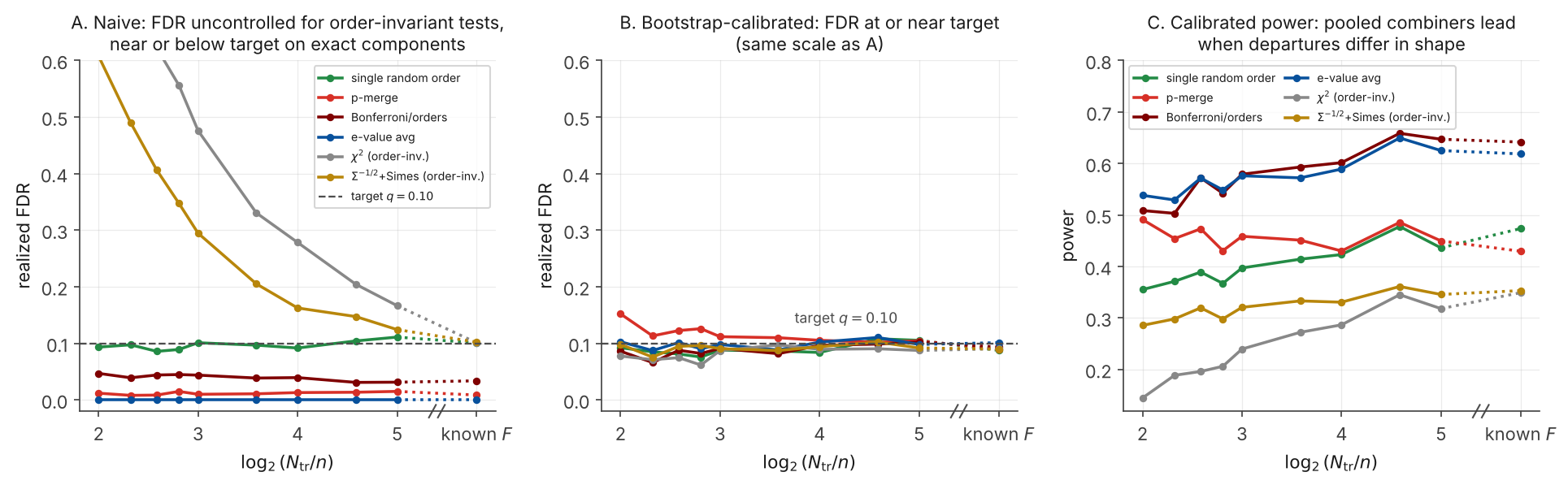}
\caption{Multiplicity layer with \emph{heterogeneous} departures ($n=20$, $\rho=0.5$, $M=12$,
$\mathrm{ncp}=20$, $200$ hypotheses, $\pi_0=0.95$, non-nulls half sparse and half dense, target
$q=0.10$, $B=2500$, $R=200$). Training-sample size is on a log scale, from $\Ntr=4n$ to $32n$, with known $F$ after the break. Under estimation the Rosenblatt-based tests use the exact
predictive scores. (A) With na\"ive per-hypothesis statistics the realized FDR is
uncontrolled for the order-invariant tests, on plug-in scores, near or below
target for the Rosenblatt-based tests, and near zero for the e-value average.
(B) After bootstrap calibration every method sits at target at every training size, apart from
$p$-merge at $\Ntr=4n$, and slightly above it from $5n$ on,
within two standard errors from $8n$ on (same scale as A). (C) Only then is a power comparison meaningful: the
pooled combiners lead at every training size, because neither order-invariant baseline
is competitive on both kinds of departure: the symmetric-root test misses the dense half, and
$\chi^2$, flat across shapes, is mediocre on both. Na\"ive power is deliberately not shown, being bought with excess FDR.}
\label{fig:estF}
\end{figure}

\subsection{Which combiner, and how many orderings}\label{sec:which}

\textbf{e-value averaging} is the combiner to reach for by default. It is never the weakest of the
three pooled rules in any setting we ran, and it is never far from the best. Figures~\ref{fig:boot}
and \ref{fig:estF} report ten training sizes each, from $\Ntr=4n$ to $32n$ and known $F$; across those twenty
settings the e-value average's largest shortfall against the leading combiner is
$0.023$, against $0.067$ for Bonferroni-over-orders and $0.212$ for $p$-merge. That is the whole
argument for it. The alternatives win in one regime and lose in the other---$p$-merge leads the
single-hypothesis comparison at $\alpha=0.05$ (Figure~\ref{fig:boot}) and trails
Bonferroni-over-orders by $0.212$ in the screen (Figure~\ref{fig:estF}); Bonferroni leads the screen at $6n$ and from $8n$ on
and is last of the three at $\alpha=0.05$---whereas the e-value average is second in thirteen of the twenty
settings and first in the remaining seven, and in the threshold sweep reported below it is first or
second at every level from $0.05$ down to $0.0005$. A combiner chosen before one knows whether the departure will
be sparse or dense, or whether the threshold will be conventional or a screen's, should be the one
without a bad case.

Choosing it costs nothing in the form of the output. Under bootstrap calibration the e-value average
is treated exactly like any other statistic---Procedure~1 returns a calibrated $p$-value for it
directly, its accuracy being that of the bootstrap itself (Section~\ref{sec:calib-exp}) plus
the $1/(B+1)$ granularity that every rank-based $p$-value shares. Under a \emph{fully specified} null, or on the
exact predictive scores under an estimated one, where $\bar E_M$ is a genuine e-value
(Lemma~\ref{lem:evalue} and the caveat after it), no calibration is needed at all: Markov's
inequality gives the valid $p$-value $1/\bar E_M$ at once. The converse does not hold: a $p$-value cannot be turned into
an e-value without a calibrator and a loss. So under a fully specified null the e-value
average yields both currencies, feeds e-BH directly for FDR control under arbitrary dependence
among the (valid) e-values (Section~\ref{sec:fdr}), and composes across independent datasets by
multiplication (or sequentially, when each factor is an e-value conditionally on the past); under
an estimated-population null the same holds on the exact predictive scores, whereas on plug-in
scores it would be a calibrated test statistic like the others.

Why the two specialists trade places is worth setting out, since it is what makes the default
safe. A ranking of combiners stated at $\alpha=0.05$ does not survive a move to the far smaller
thresholds a screen imposes. Benjamini--Hochberg over $N$ hypotheses reports a rejection
only when some $p$-value reaches $qk/N$, which for the screen of \S\ref{sec:fdr-main} means
$0.0005$ to $0.005$ when the rejection set consists of the non-nulls (false positives raise
$k$ and loosen the rung): one to two orders of magnitude below the level of a single test. The
combiners are not equally placed to supply such values. Bonferroni-over-orders, $M\min_m P^\sigma_m$,
keys on the single most favourable ordering, so its calibrated $p$-value has a heavier left tail
than the averaging rules and a worse typical value; $p$-merge and the e-value average use all $M$
views, which pays at a moderate threshold and costs at an extreme one. Holding the design of
Figure~\ref{fig:boot} completely fixed---same covariance, dimension, alternative (two shifted
coordinates, $\mathrm{ncp}=12$) and $M$, known $F$---and sweeping only $\alpha$ makes the mechanism
plain (the sweep is a separate, larger Monte-Carlo run with one fixed set of orderings and
$2\times10^5$ draws, so its $\alpha=0.05$ values differ a little from Figure~\ref{fig:boot}'s):
$p$-merge leads Bonferroni-over-orders at $\alpha=0.05$ ($0.510$ against $0.487$) and trails
it at $\alpha=0.0005$ ($0.045$ against $0.066$), having lost over nine-tenths of its power to
Bonferroni's six-sevenths, with nothing but the threshold changing. Against the dense half of the
Figure~\ref{fig:estF} screen (all coordinates shifted, $\mathrm{ncp}=20$) the same sweep is
starker: $p$-merge falls from $0.757$ to $0.094$ while Bonferroni falls only from $0.892$ to
$0.393$, their gap widening from $0.135$ to $0.299$. This is why Bonferroni-over-orders is the weakest of the three pooled
combiners in Figure~\ref{fig:boot} and first in Figure~\ref{fig:estF}. It is also why the e-value average is the
safe default: alone among the three it is competitive at both ends of the threshold range, so it
does not require the user to know in advance which end they are at. If that is known, the
specialists are slightly better---$p$-merge for a single hypothesis at a conventional level,
Bonferroni-over-orders in a sparse screen.

\paragraph{How many orderings}
Proposition~\ref{prop:finiteM} supplies the reasoning and simulation supplies the numbers. Under
the alternatives of Figures~\ref{fig:boot} and~\ref{fig:power}, and in the sweep design below, the
average spread $\bar V$ of the per-ordering e-values lies between $0.24$ and $1.16$ (under the null
at $n=10$, $\rho=0.5$ it is $0.044$), so the proposition predicts that twelve orderings leave about
$\bar V/24$, that is $0.010$ to $0.048$ nats, uncollected. Measured against $2048$ orderings as the
reference, they leave $0.009$ to $0.054$ nats and collect $86\%$ to $92\%$ of the incremental
expected log-evidence gain from averaging, which is $0.09$ to $0.39$ nats in all. Fixed-level
power levels off sooner. With calibrated thresholds at $n=10$, $\rho=0.6$, one shifted coordinate
and known $F$, and a fresh set of orderings in each of $400$ realizations, power rises with $M$
and then levels off: from a single random ordering to $M=12$ it gains $0.049$ for $p$-merge,
$0.038$ for the e-value average and $0.034$ for Bonferroni-over-orders, and going on to $M=128$
changes it by at most $0.005$ (paired standard errors $0.001$). At $\rho=0$ every combiner is
flat, as Proposition~\ref{prop:main}(iii) predicts. The decision itself settles more slowly. Two
independent sets of twelve orderings, calibrated on the same null draws so that only the orderings differ, reach different calibrated
decisions of the e-value average for $2.9\%$ of observations under the one-coordinate alternative
above and $6.2\%$ under the all-coordinate alternative of Figure~\ref{fig:power}, and for $1.3\%$ under the null in both designs, against $12.1\%$ and $18.4\%$ under the alternatives with a single
ordering. With $128$ orderings the four figures fall to $1.0\%$,
$1.8\%$, $0.5\%$ and $0.4\%$, and they are still falling, at about the rate $M^{-1/2}$: observations
near the threshold need many orderings to settle (Corollary~\ref{cor:stable}), so at twelve
orderings the seed of the orderings belongs with a reported decision. Calibrating each set on its
own null draws, as an analyst repeating the whole procedure would, brings in a second source of
variability, the calibration sample: two independent calibrations of one continuous statistic
disagree on about $0.55\%$ of null observations at $B=1999$. The two sources do not add, and the
four figures at twelve orderings become $3.3\%$, $6.2\%$, $1.4\%$ and $1.3\%$. In these designs a dozen orderings, the value
used throughout the paper, is where calibrated power levels off, and it collects about nine tenths
of the incremental expected log-evidence gain from averaging, at a cost linear in $M$ ($BM$
whitenings under bootstrap calibration). These are results for the designs examined, not a
guarantee: calibration absorbs Bonferroni's $\times M$ penalty, but changing $M$ still changes the
statistic, and at na\"ive thresholds Bonferroni-over-orders has an interior optimum
(Proposition~\ref{prop:optM}, Remark~\ref{rem:dome}).

\subsection{Comparison with conformal novelty detection}\label{sec:conformal}

A reader who screens many multivariate observations for anomalies will ask how this compares with
conformal novelty detection, which has developed rapidly in exactly this space
\cite{bates2023,marandon2024}. The honest answer begins with a distinction, because the two
literatures test different nulls.

Our null is \emph{conformity to a specified $F$}: $H_0^{(j)}:\X^{(j)}\sim F$, with $F$ supplied by
the modeller and, when unknown, estimated from a training sample and accounted for by
the exact predictive scores and Procedure~1. The conformal null is \emph{exchangeability with a reference sample}: given inliers
$Y_1,\dots,Y_{\Ntr}$ and test points $\X^{(1)},\dots,\X^{(N)}$, one computes a nonconformity score
$S$, holds out $\ell$ of the reference points for calibration, and forms
\begin{equation}\label{eq:confp}
p_j \;=\; \frac{1+\#\{i \in \text{cal}: S_i \ge S(\X^{(j)})\}}{\ell+1},
\end{equation}
which is super-uniform under exchangeability alone. Under the independence and scoring
assumptions of Bates et al.\ \cite{bates2023}---test inliers and reference data jointly
independent, with conditions on the score distribution, which our simulation design
satisfies---these $p$-values are PRDS, so BH controls FDR; Marandon et al.\ \cite{marandon2024} observe that a
prespecified score wastes power and propose \emph{AdaDetect}, which splits the reference sample,
trains a probabilistic classifier to separate the fit part from the calibration part pooled with the
test sample, and uses the predicted probability of belonging to the latter as $S$. Because that
score is invariant to permutations within calibration-plus-test, \eqref{eq:confp} remains valid, so
the score may be \emph{learned} without any distributional assumption. Our implementation
makes that invariance literal: each classifier is trained on the calibration and test points in a
canonical (lexicographic) order, so the fitted score depends on them only through their unordered
set. Without that step a random forest with a fixed seed would not qualify, since it draws its
bootstrap samples by row position.

The trade is therefore explicit. By conformal methods we mean, as in the works cited,
procedures calibrated against observed inliers alone (a rank $p$-value of the same form computed
against draws \emph{simulated} from a specified model is the Monte-Carlo test of Procedure~1, and
belongs to the parametric route). They buy validity without distributional assumptions, at the
price of needing a reference sample at test time, and they answer a different question from ours:
whether a new observation resembles the reference data, not whether a given model is correct. If a
model misdescribes the reference data and the new data in the same way, the new data remain
exchangeable with the reference sample and a conformal test has nothing to detect, whereas the
parametric route asks directly the question a risk manager asks---is \emph{this} model right?
Because the two routes test different nulls, their power can be compared only where the nulls
coincide, that is, when the specified model is correct, and there one should expect access to the
model to help. The comparison below therefore asks not \emph{whether} the model helps but
\emph{how much}: how large the advantage is once the model must be estimated from the same
reference sample the conformal methods receive, in the sparse regime in which screens operate, and
whether a learned conformal score closes the gap. We do not study misspecification, where the two
routes answer different questions and power is not the relevant comparison.

\paragraph{Design} We reuse the screening problem above ($n=20$, $N=200$ hypotheses,
$\mathrm{ncp}=20$, half sparse and half dense non-nulls, $\rho=0.5$, $q=0.10$, $M=12$, $B=2500$),
supply both routes with the same reference sample, and compare five conformal scores: the Mahalanobis
form $\mathbf{y}^\top\hat\Sig_{\mathrm{fit}}^{-1}\mathbf{y}$, the nearest-neighbour distance, the
negative log of a Gaussian kernel density estimate, and AdaDetect with a random forest and with
penalized logistic regression on quadratic features.

\paragraph{Held-out calibration size} The resolution consideration behind $B>N/q$ on the
bootstrap binds far harder on the conformal calibration set once novelties are rare. A conformal
$p$-value is a multiple of $1/(\ell+1)$, so reaching the $m_1$-th BH threshold $qm_1/N$
requires $\ell+1\ge N/(qm_1)$, with $m_1$ a lower bound on the number of rejections
(Remark~2.2 of \cite{marandon2024} states the condition in the strict form $\ell>N/(qm_1)$); with
the ten non-nulls of \S\ref{sec:fdr-main} this is $\ell\ge199$, strictly $\ell>200$. We start the
reference-sample grid at $\Ntr=400$ and use a balanced split $\ell=\Ntr/2$, which clears the
condition at $\Ntr=1000$ and $2000$; at $\Ntr=400$ the split $\ell=200$ sits exactly on the
resolution boundary, where any nonempty BH rejection set contains at least ten hypotheses (if
no true null is rejected, all ten non-nulls must be rejected together), so the $\Ntr=400$
conformal column of Table~\ref{tab:conformal} is resolution-limited by construction and should be
read as such. Where a method's authors recommend a value we use theirs: for AdaDetect
the same remark recommends $\ell=N$, which we adopt wherever its strict condition admits it
and replace by the balanced split where it does not (the balanced split is a design choice
fixed in advance; the resolution condition only bounds $\ell$ from below). All figures are averages over $R=200$
replications, with Monte-Carlo standard errors given in the captions.

\begin{table}[t]\centering\small\setlength{\tabcolsep}{4pt}
\caption{Comparison with conformal novelty detection in the sparse screening regime
($n=20$, $N=200$ hypotheses, $\pi_0=0.95$, so ten novelties, $\mathrm{ncp}=20$, $q=0.10$,
$R=200$). Entries are calibrated power. In these simulations no
estimated FDR exceeds the target of $0.10$ by more than one Monte-Carlo standard error; the largest
estimate is $0.108$. For the bootstrap-calibrated tests this is evidence of control, not a proof of
it (\S\ref{sec:fdr-main}). The pooled tests
use the exact predictive scores. Standard errors are at
most $0.012$ on the realized FDR and $0.018$ on the powers. Every tuning constant is fixed in
advance and none is selected on the results: the distance to the nearest neighbour as the
nonparametric score, and a balanced calibration split $\ell=\Ntr/2$, a design choice fixed
in advance (the resolution condition only bounds $\ell$ from below). The data are generated from
the specified Gaussian model, the setting in which the two routes test the same null.}
\label{tab:conformal}
\begin{tabular}{lccc}
\toprule
& $\Ntr=400$ & $\Ntr=1000$ & $\Ntr=2000$\\
\midrule
Pooled, Bonferroni-over-orders (this paper) & \textbf{0.642} & \textbf{0.648} & \textbf{0.666}\\
Pooled, e-value average (this paper)        & 0.629 & 0.620 & 0.637\\
Pooled, $p$-merge (this paper)              & 0.497 & 0.454 & 0.468\\
$\chi^2$ energy test                        & 0.336 & 0.338 & 0.369\\
\addlinespace
Conformal, nearest-neighbour distance       & 0.073 & 0.405 & 0.520\\
Conformal, Mahalanobis                      & 0.038 & 0.183 & 0.336\\
Conformal, kernel density                   & 0.037 & 0.141 & 0.265\\
AdaDetect, random forest                    & 0.003 & 0.108 & 0.218\\
AdaDetect, penalized logistic               & 0.003 & 0.026 & 0.063\\
\bottomrule
\end{tabular}
\end{table}

\paragraph{What the specified model is worth when it is right}
Table~\ref{tab:conformal} answers the quantitative question. With the model estimated from the same
reference sample, the pooled tests lead across the whole range of reference-sample sizes: at
$\Ntr=1000$---five times the size of the screen itself---Bonferroni-over-orders detects $0.648$ of
the ten novelties against $0.405$ for the best conformal method, and at $\Ntr=2000$, ten times the
screen, $0.666$ against $0.520$, gaps of $11.0$ and $8.7$ Monte-Carlo standard errors. The rows
compare complete procedures, each a score together with a calibration, and the specified model
enters the pooled tests through both: through the Rosenblatt score, and through a calibration that
simulates from the fitted model rather than holding out inliers. One pair of rows uses the same
quadratic-score form. The conformal Mahalanobis score and the $\chi^2$ energy test are the same quadratic form (with
the covariance estimated from the fitting part of the reference sample in the conformal case),
calibrated respectively against held-out inliers and by Procedure~1, so the gap between
them---$0.038$ against $0.336$ at $\Ntr=400$, $0.183$ against $0.338$ at $1000$, $0.336$ against
$0.369$ at $2000$---illustrates the combined benefit of model-based calibration and of retaining
the whole reference sample for fitting (the conformal version fits on half of it and calibrates
on the other half, with the resolution of a held-out set): decisive when reference data are
scarce and small once they are plentiful. The gaps narrow over the reference-sample sizes
examined; performance beyond $\Ntr=2000$ was not assessed, and different scores need not converge
to equal power (the conformal Mahalanobis score approaches an energy test, which at known $F$
remains well below the pooled combiners, Figure~\ref{fig:estF}). The reference samples of the
larger sizes are in any case larger than the approximately $250$ observations used in our FX
application, which calibrates on a single year of returns and tests the next.

It is worth also comparing the conformal methods with one another. AdaDetect exists because a
prespecified nonconformity score wastes power, and it replaces one with a score learned from the
data; here that reverses. At $\Ntr=2000$ its random forest reaches $0.218$ and its penalized
logistic version $0.063$, while the three prespecified scores reach $0.520$, $0.336$ and $0.265$:
the two learned scores are the weakest of the five. This is not a criticism of AdaDetect, which is
designed for a regime it is not being given here. A classifier trained to separate inliers from a
test sample can only learn what the novelties in that test sample reveal, and ten novelties diluted
among a hundred and ninety inliers reveal very little.

\section{Application: validating a Gaussian foreign-exchange risk model}\label{sec:fx}

We close the empirical work with a real risk-management application in which the coordinate ordering
is genuinely arbitrary: daily foreign-exchange (FX) returns. The FX market is where currencies are
traded; an \emph{exchange rate} is the price of one currency in another (e.g.\ dollars per euro),
and a currency's daily \emph{return} is the relative change in that price from one trading day to
the next. A bank or fund holding positions in several currencies summarizes a day of market movement
as one vector $\X_t\in\R^n$---the day's returns of the $n$ currencies it is exposed to---and its
\emph{risk model} is a probability distribution for that vector, used to set capital buffers and
loss limits. The simplest and still most widely used such model is multivariate Gaussian with mean
and covariance estimated from recent history. \emph{Validating} the model means asking, day after
day, exactly this paper's question: is today's observed vector consistent with the model's
distribution? In the language of the forecasting literature, this is multivariate density-forecast
evaluation \cite{dgt1998,dovern2020}: the fitted $\mathcal N(\hat\mut,\hat\Sig)$ is a one-day-ahead
density forecast, held fixed within each year, and each day's realization is scored against it. And the coordinates of that vector are currencies: no convention orders them, so any
fixed choice (alphabetical tickers, portfolio weights, legacy code) is exactly the hidden analyst
degree of freedom this paper is about.

\paragraph{Setup}
Data are daily spot exchange rates for nine major currencies against the U.S.\ dollar (euro,
Japanese yen, British pound, Canadian dollar, Swiss franc, Australian and New Zealand dollars,
Swedish and Norwegian kronor; all expressed as USD per foreign unit), published by the Federal
Reserve (H.10 release, via FRED), 2015--2025: $2{,}747$ days of log-returns with $n=9$. The model to
validate is refitted in the standard \emph{walk-forward} way: at the start of each calendar year,
$\hat\mut$ and $\hat\Sig$ are estimated from the preceding year's $\approx250$ trading days (so
$\Ntr\approx28n$---moderate, realistic estimation), and every day of the new year is tested against
$\mathcal N(\hat\mut,\hat\Sig)$ out of sample. The null tested each day deserves a precise
statement, because two readings are possible and they license different guarantees. We test the
\emph{issued forecast}, stated conditionally: $H_0^{(t)}:\X_t\mid\mathcal F_{t-1}\sim\mathcal
N(\hat\mut_t,\hat\Sig_t)$, with $(\hat\mut_t,\hat\Sig_t)$ the numbers the model actually produces
from information available before day $t$---the density-forecast-evaluation
null of \cite{dgt1998}. Under this null the whitened scores are \emph{exactly} standard normal,
Lemma~\ref{lem:evalue} applies as stated, and the raw pooled e-values are valid without
adjustment; rejections mean ``the forecast, as issued, is wrong about this day.'' This differs
from the estimated-population null of \S\ref{sec:calib}, under which the day's vector and the
preceding year's returns are independent draws from one Gaussian population. There the plug-in scores are not
e-values (see the caveat after Lemma~\ref{lem:evalue}), so the Rosenblatt-based tests use the exact predictive scores of
Lemma~\ref{lem:cpit} with the mean estimated, on which the pooled e-values are valid under that null
too. The two nulls call for different calibrations, and we run both. Under the
issued-forecast null the whitened scores are exactly standard normal, so each statistic's null
distribution is obtained by simulating from the issued forecast itself and whitening by the same
fixed factors as the data (Procedure~1 with step~1 replaced by exact draws from
$\mathcal N(\hat\mut_t,\hat\Sig_t)$; no re-estimation). Under the estimated-population null we use
the re-estimating bootstrap exactly as in Section~\ref{sec:exp}, with mean and covariance
re-estimated inside every replicate. The per-year rejection rates below are reported under both, and so
is the e-BH screen; the e-values quoted are those of the issued forecasts. For each
test year $2016$--$2025$ we run, at calibrated level $\alpha=0.05$ ($B=999$, $M=12$ fresh
orderings per day): the single-random-ordering Simes test, the three pooled combiners, and the two
order-invariant references. Every day is one multivariate hypothesis---the paper's basic setting---and the ten years
form a $2{,}497$-day screening problem for e-BH. If the Gaussian model were exactly right, every
method would flag $\approx5\%$ of days each year; systematic exceedance of $5\%$ is evidence of model
failure, and \emph{which} days carry the evidence is exactly what a risk manager wants to know.

\paragraph{The test behaves as designed}
In calm years the calibrated rejection rates sit near the nominal $5\%$ (2017: $0.06$; 2021: $0.03$;
2023: $0.06$; 2024: $0.05$ for the e-value average under the estimated-population calibration;
$0.06$, $0.03$, $0.08$, $0.05$ under the issued-forecast calibration, which absorbs no
estimation noise and so runs at or a little above the other in every year). One caution in reading this: with the truth
unknown, a near-nominal rejection rate is \emph{consistent with} adequacy rather than proof of
it---a misspecified family can produce the same rate for this statistic. What the calm years do
establish is that the calibrated procedure behaves sensibly out of sample: last year's
covariance is not visibly contradicted by this year's ordinary days. In stress years the model is flagged
massively: $30\%$ of days in 2020 (the COVID pandemic), $49\%$ in 2022 (the U.S.\ Federal Reserve's
aggressive rate hikes and the resulting dollar surge), $16\%$ in 2025 (the April tariff shock),
$14\%$ in 2016 (the United Kingdom's Brexit referendum); $32\%$, $55\%$, $17\%$ and $14\%$
under the issued-forecast calibration. The pooled e-values also \emph{rank} the
history interpretably (we quote the issued-forecast e-values): an e-value is the factor by which the day's evidence multiplies a bet
against the model, so $\bar E_M=20$ is borderline ($=1/\alpha$) while
$\bar E_M=10^{13}$ is overwhelming. The largest values in ten years, in order, are the Brexit
referendum result day 2016-06-24 ($\bar E_M\approx10^{20}$: sterling fell $8\%$ against a model
fitted on placid 2015), the Bank of Japan's surprise loosening of its yield-curve control on
2022-12-20 ($\approx5\times10^{17}$: a yen-specific jump), the COVID ``dash-for-dollars'' days
2020-03-18/20/23 ($5\times10^{12}$--$10^{14}$), the day the Swiss National Bank unexpectedly raised rates
alongside a large Fed hike, 2022-06-16, and the 2025-04-04 tariff announcement
(Figure~\ref{fig:fx}A). None of these dates was supplied to the procedure; they emerge from the
statistics alone. The exact predictive e-values, valid under the estimated-population null,
are smaller, because they allow for estimation error ($5\times10^{13}$ for the referendum day), and
rank the days almost identically (rank correlation above $0.9999$; the same eight days lead). Finally, applying na\"ive e-BH at $q=0.10$ across all $2{,}497$
days---consecutive days are serially dependent, which e-BH tolerates by construction, no
adjustment needed---flags $78$ days, essentially the union of these episodes: a defensible list
of ``days on which the issued forecasts demonstrably failed,'' with FDR control that is valid
under the forecast nulls. On the exact predictive e-values, e-BH flags $62$ days, all
among the $78$: the more cautious list, with FDR control under the estimated-population null, of
days whose returns are inconsistent with coming from the same Gaussian population as the preceding
year's, whatever its parameters.

\paragraph{What pooling buys here---and honestly, what it does not}
The day-level paired gain of the e-value average over a single random ordering is $+0.006\pm0.004$ under the estimated-population calibration and
$+0.004\pm0.004$ under the issued-forecast one: statistically indistinguishable from zero (the days are serially dependent, so the standard
error is computed with a Newey--West adjustment; the conclusion of a null gain is unchanged
either way). This is the operating map of \S\ref{sec:headline} speaking,
not a failure: real FX stress days are \emph{far} from the borderline-power regime---when the
departure is a $10$-sigma event every ordering rejects, and when the day is ordinary none does. The
practical value of pooling in this application is different and threefold: it removes the arbitrary
ordering from the \emph{definition} of the procedure (the reported statistic no longer depends
on a ticker convention: at finite $M$ it is a Monte-Carlo approximation to the fully
order-symmetric statistic, whose distribution is the same whichever labelling the data arrive in;
the realized number is reproducible for a given input labelling and seed, and exact invariance
holds in the $M\to\infty$ average), it supplies valid 
e-values for dependence-robust day screening via e-BH under either null, and its observed rejection rates are
indistinguishable from a single ordering's---which, without ground-truth labels for the days, is a
statement about how often the methods flag, not a demonstration of equal power. The
episode-by-episode pattern of \S\ref{sec:headline} also appears here, with the same caveat that
these are observed flag rates rather than measured power: the
best-performing reference \emph{changes across episodes} (the symmetric-root test leads in 2020,
$\chi^2$ leads in 2022; Figure~\ref{fig:fx}B), while the pooled test tracks within
$0.01$--$0.05$ of whichever is winning without knowing the episode's shape in advance.

\begin{figure}[tp]\centering
\includegraphics[width=\linewidth]{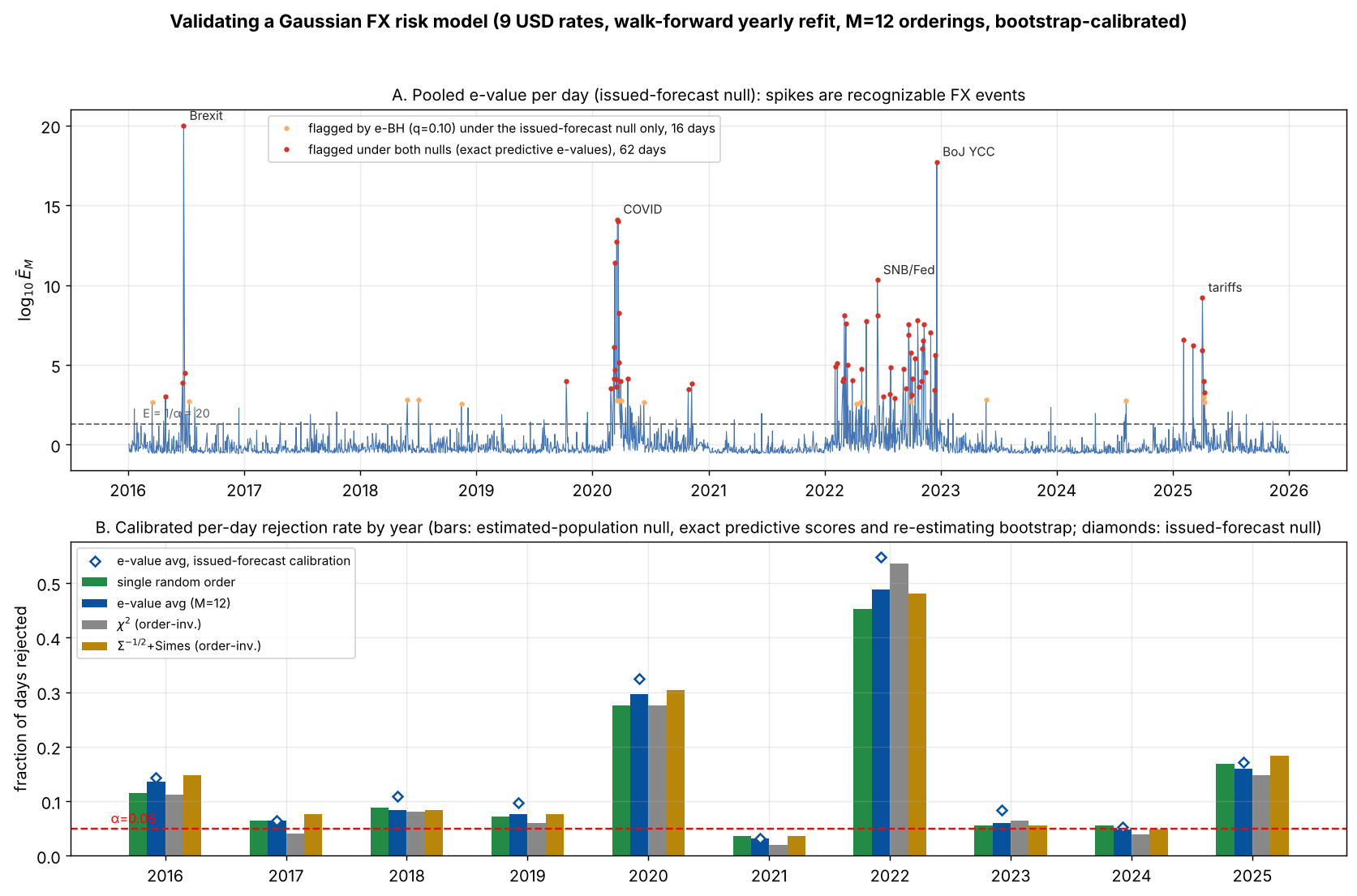}
\caption{FX application (9 currencies vs.\ USD, 2016--2025, yearly walk-forward refit). (A) Pooled
e-value per day under the issued-forecast null; dots are the $78$ days flagged
by e-BH at $q=0.10$, in red the $62$ also flagged on the exact predictive e-values, which are valid
under the estimated-population null; annotated spikes are the
Brexit referendum, the COVID dollar scramble, the SNB/Fed and BoJ surprises of 2022, and the 2025
tariff shock. (B) Calibrated per-day rejection rates by year (bars: estimated-population null,
exact predictive scores and re-estimating bootstrap; diamonds: e-value average under the issued-forecast null, exact
simulation from the issued forecast): near-nominal in calm years, large in
stress years; the leading order-invariant reference changes across episodes while the pooled test
stays near the top.}
\label{fig:fx}
\end{figure}

\section{Discussion}\label{sec:disc}

The paper followed a single question---whether one must gamble on a coordinate ordering---to a
three-part answer: reordering only reshuffles a fixed signal, so pooling \emph{can} help; the base
test must look in both tails, or the gamble is the least of one's problems; and calibration is what
converts the latent advantage into a usable one. We close by stating what this buys, where it stops,
and how it connects.

\paragraph{What the evidence supports}
The reason to pool is not that any one ordering is bad but that one cannot know in advance which
ordering, or which fixed whitening, suits the departure at hand. That is why the gain over a single
random ordering is real yet modest, while the gain over an order-invariant alternative is
large wherever that alternative is weak---the symmetric-root test on diffuse departures, the
$\chi^2$ test, whose power does not depend on the shape, on concentrated ones---and, in a screen whose
hypotheses deviate from the null in different ways, decisive. The theory says the same thing in
another language: reordering conserves the signal energy and only redistributes it
(Lemma~\ref{lem:energy}), so pooling cannot manufacture evidence, it can only stop one from
betting everything on a single view of it (Proposition~\ref{prop:main}). The foreign-exchange
application shows the machinery surviving contact with data it was not designed around: rejection
rates stay near nominal out of sample in calm years (consistent with adequacy, not proof of it),
the largest e-values are the episodes a risk manager would name
unprompted, and e-BH screens a decade of serially dependent days without adjustment.

\paragraph{An honest operating map}
The boost over a single ordering is small at near-independence (small orbit), and it is masked
entirely by na\"ive thresholds---calibration is not optional. The symmetric-root test is a
competitive order-invariant alternative for sparse departures, and it collapses for dense ones
(Figure~\ref{fig:power}); the $\chi^2$ test does not depend on the shape at all, but trails the
pooled tests on sparse departures. The case for pooling is that, at known $F$ and, under estimation, with the exact predictive scores from the smallest training sample we ran (within $0.006$ of the symmetric-root test on a single shifted coordinate) (Table~\ref{tab:estgrid}), the e-value average
matches or beats both against concentrated and fully diffuse departures without needing to know
which it faces. One methodological point deserves emphasis:
pooling must be evaluated against the \emph{expected} power of a random ordering, not against one
fixed ordering, whose value depends on which coordinates happen to precede the departure and can
overstate the gain several-fold.
Guidance on the choice of combiner is collected in \S\ref{sec:which} and not repeated here; the
short version is that the e-value average is the default, being the only one of the three pooled
rules without a regime in which it does badly, and that it costs nothing to adopt since bootstrap
calibration returns a $p$-value for it directly. The one prohibition worth restating is that the
arithmetic mean of $p$-values must not be used without calibration. To that map
\S\ref{sec:conformal} adds a boundary against conformal novelty detection. A learned conformal score
must learn the alternative from the novelties present in the test sample, and where those are rare
it has little to learn from; since screens are usually sparse, that leaves a good deal of ground to
a procedure that brings its own null, when that null is right.

\paragraph{Limitations and scope}
We deliberately restrict to the Gaussian null. Behavior under heavier-tailed and non-parametric nulls
is more intricate---the orbit geometry, calibration quality, and even the sign of the gain can change
with the tail behavior of $F$---and a full treatment is left to future work. It bears saying plainly
that the error control established here is contingent on that family being the right one: like every
test of a specified null, and unlike the conformal procedures of \S\ref{sec:conformal}, our
guarantee is a statement about departures from $F$ and not about departures from whatever generated
the data. Our calibration
machinery is general (Procedure~1 needs only the ability to simulate from $\hat F$ and recompute the
transform), so extending to copula and non-parametric nulls, where the orbit may be even richer, is a
natural next step. The exact predictive whitening of \S\ref{sec:calib}, by contrast, is specific to
the Gaussian family; conditional probability integral transforms exist for some other parametric
families \cite{oreilly1973}, and where none is available the plug-in transform with Procedure~1
remains the general route. Our power study uses location (mean-shift) alternatives at fixed energy under
equicorrelated covariances (with a hub design used only for the nominal-threshold sweep of
Remark~\ref{rem:dome}); broader covariance designs, scale/dependence/shape alternatives,
estimated nuisance parameters outside the Gaussian family (its mean and covariance are
handled exactly, \S\ref{sec:calib}), and skewed nulls remain open. The location-style Simes base statistic
is also not omnibus: for instance, an \emph{under}-dispersed coordinate ($\sigma=0.5$ against an
assumed $1$) leaves the power of the one- and two-sided versions alike at or below the nominal
level (about $9\times10^{-5}$ for the two-sided test at $\alpha=0.05$ in the one-dimensional
case), so scale-aware
base statistics are a natural addition to the toolkit.

\paragraph{Relation to prior work}
The closest predecessor is \cite{dovern2020}, who construct order-invariant tests from
conditional PITs across orderings and address estimation uncertainty; our full-order average
$\E_\sigma[E^\sigma\mid\X]$ is of course itself a symmetrization, so the distinction is not
aggregation versus symmetrization but \emph{what} is aggregated and with what guarantee. What
this paper adds over that construction is: aggregation at the level of \emph{evidence}
(e-values and $p$-values) with finite-$M$ validity under the arbitrary dependence that
recomputation creates (Lemma~\ref{lem:merge})---a randomized $M$-ordering test is valid as a
test, not only in a symmetrized limit; the two-sided base statistic, which removes the
directional blind spot of one-sided scoring when the sign of the departure is unknown
(Section~\ref{sec:misspec}); the classical exact predictive whitening
\cite{oreilly1973,rincon1979} applied under every ordering, with the re-estimating calibration that makes
the aggregate usable and powerful under estimation; and the multiplicity layer. Specification testing via the
Rosenblatt transform is classical for copulas \cite{genest2009} and for density-forecast evaluation
\cite{dgt1998}. The e-value averaging used to combine orderings, and the e-BH multiplicity layer of
Section~\ref{sec:fdr}, connect the test to recent e-value theory
\cite{vovkwang2021,vovkwang2020,wangramdas2022,grunwald2024}. Contemporary
PIT-of-order-statistics methods combine the transform with Cauchy combination \cite{covington2025}; our
contribution is orthogonal---the ordering-aggregation geometry, its calibration, and the power
characterization. The closest recent alternative for the screening problem of
Section~\ref{sec:fdr} is conformal novelty detection \cite{bates2023} and its adaptive extension
\cite{marandon2024}, which replace a specified null by exchangeability with a reference sample;
\S\ref{sec:conformal} sets out the different hypotheses the two answer and measures what the
specified model is worth when it is correct.

\paragraph{What this paper is not}
It is worth forestalling a natural misreading: the recipe here is not ``whenever several tests
exist, apply all of them and merge the evidence.'' Validity is not what separates the two:
Jensen's inequality and the merging lemmas apply to any collection of valid e-values. What the
generic recipe lacks is the structure that makes the average worth taking here, and this paper
itself furnishes the counterexamples. The $M$ views pooled here are not
different methods but the \emph{same} statistic under a symmetry of the null, and that structure is
what the results use: each view carries exactly the same total signal (Lemma~\ref{lem:energy}), so
averaging is over equally informed presentations, never over procedures of unequal merit; and the
choice being averaged over is \emph{forced}---the transform cannot be computed without an ordering---so
the alternative to pooling is not ``pick the best method'' but ``commit to an arbitrary one.''
Where those conditions fail, so does the slogan. Pooling adds nothing to the Fisher base statistic
and pooled Fisher can fall below a single ordering (Remark~\ref{rem:fisher}); against the
exchangeable dense alternative every ordering carries the identical \emph{mean} profile
($P_\sigma\mut=\mut$, with an exchangeable $\Sig$), so there is no ordering lottery for
randomization to fix and any pooling gain there operates through the noise rather than the
signal; and without calibration the multiplicity cost of merging cancels the gain
almost exactly (Remark~\ref{rem:dome}). What the paper actually argues is narrower: when a
\emph{nuisance choice} indexes equally informed views of one observation, do not commit---average,
and calibrate so that the averaging is free.

\paragraph{Acknowledgments}
Claude Opus 4.7 and Claude Fable 5 (Anthropic) were used in the preparation of this manuscript; the
author thanks Javad Alipanah for providing access to them. The author retains full responsibility
for the mathematical content and conclusions of the paper.

\paragraph{Declaration of generative AI in scientific writing}
During the preparation of this work the author used Claude (Anthropic) in the design and
verification of the experiments, in drafting, and in checking the mathematical arguments. After
using this tool, the author reviewed and edited the content as needed and takes full
responsibility for the content of the published article.

\appendix
\appendixsetup
\section{Reproducibility}\label{app:repro}
All experiments are Monte-Carlo simulations with fixed integer seeds, so every number, table and
figure in the paper can be regenerated exactly; we deliberately avoid process-dependent sources of
randomness. The implementation matches the paper one-to-one: the Gaussian Rosenblatt map is the
Cholesky whitening \eqref{eq:whiten}; the base statistics are those of \S\ref{sec:base}; the
combiners are those of \S\ref{sec:combine}; and every calibrated threshold is produced by
Procedure~1. Under an estimated null the Rosenblatt-based tests use the exact predictive
scores \eqref{eq:exactz}; the scripts compute them alongside the plug-in scores on the same draws,
so the plug-in figures quoted for comparison come from the same simulations.

Three points of experimental design bear repeating, since they affect how the numbers should be
read. First, the $M$ orderings are redrawn in every Monte-Carlo realization in every
experiment behind a figure or table (the threshold sweep of \S\ref{sec:which} and the supporting
checks behind Remarks~\ref{rem:bridge} and~\ref{rem:dome} hold one fixed bank of orderings, as
stated there), so the single-ordering
baseline estimates the \emph{expected} behaviour of an arbitrarily chosen ordering rather than the
luck of one fixed permutation. Second, the order-invariant reference tests ($\chi^2$ and
$\hat\Sig^{-1/2}$+Simes) are bootstrap-calibrated by the same Procedure~1 as the Rosenblatt-based
tests, so all power comparisons are made at matched size (exactly at known $F$, approximately
under estimation; \S\ref{sec:calib-exp} reports the realized sizes). Third, in the multiplicity layer the
bootstrap replicate count should satisfy $B>N/q$: a bootstrap $p$-value cannot fall below $1/(B+1)$,
while the smallest Benjamini--Hochberg threshold is $q/N$ (the condition is sufficient for
full BH resolution---every rung of the ladder reachable---and not necessary for making
rejections, whose thresholds $qk/N$ are larger).

The foreign-exchange application of Section~\ref{sec:fx} uses the Federal Reserve's H.10 daily
noon exchange rates for nine currencies against the U.S.\ dollar, retrieved from the FRED database
maintained by the Federal Reserve Bank of St.\ Louis; the series identifiers, quotation conventions
and retrieval date are listed in the README of the code repository below. Complete code reproducing every figure, together with the
data-retrieval script, is available at

\smallskip
\noindent\hspace*{1em}{\footnotesize\url{https://github.com/mehrdad-pournaderi/Calibrated-Order-Randomized-Rosenblatt-Tests}}

\section{Power under estimation, by departure shape}\label{app:estgrid}
\setcounter{table}{0}
Section~\ref{sec:headline} compares the tests at known $F$. Table~\ref{tab:estgrid} repeats the
comparison under estimation as the training sample grows, for three single-hypothesis departures
of equal energy: two same-sign shifted coordinates ($n=20$, the design of Figure~\ref{fig:boot}),
and one shifted coordinate or all coordinates shifted equally ($n=10$, the designs of
Figure~\ref{fig:power}). Each entry is the calibrated power of the e-value average minus that of
the stronger order-invariant baseline: the symmetric-root test for the first two departures (tied
with $\chi^2$ at known $F$ for the first) and $\chi^2$ for the third. The first entry of each pair uses the exact predictive scores of
\S\ref{sec:calib}, as everywhere in the paper; the second uses the plug-in scores of the same
draws. With the exact scores the e-value average leads at every training size on the
two-coordinate departure ($+0.021$ to $+0.037$) and on the fully diffuse one ($+0.027$ to $+0.105$,
the lead largest where estimation costs $\chi^2$ most), and on the single shifted coordinate it
trails the symmetric-root test by at most $0.006$ up to $8n$ and is level with it or ahead from
$10n$ on. With the plug-in scores the baseline led at small training sizes, by up to $0.068$ at $4n$
on the diffuse departure, and the e-value average drew level only between $5n$ and $12n$; the
calibration being the same, the difference is the whitening alone. In the screen of Figure~\ref{fig:estF}, where half the non-nulls fail on one
coordinate and half on all, the pooled tests lead at every training size (the e-value average by $0.252$ over the
symmetric-root test at $\Ntr=4n$), since neither baseline is competitive on both halves.

\begin{table}[h]\centering\small
\caption{Calibrated power of the e-value average minus that of the stronger order-invariant
baseline under estimation, by training-sample size, with the e-value average on the exact predictive
scores (``exact'') and, for comparison, on the plug-in scores of the same draws (``plug-in'');
Monte-Carlo standard errors of the paired differences are $0.002$ to $0.004$. Single hypothesis,
$\alpha=0.05$, $\rho=0.5$, $\mathrm{ncp}=12$, $M=12$, $B=199$, $R=300$. Baseline: the
symmetric-root test for the first two departures, $\chi^2$ for the third. These runs are independent
of those behind Figures~\ref{fig:boot} and~\ref{fig:power}, so entries at matching training sizes
agree with the figures to within Monte-Carlo error.}
\label{tab:estgrid}
\begin{tabular}{lcccccc}
\toprule
& \multicolumn{2}{c}{Two coordinates} & \multicolumn{2}{c}{One coordinate}
& \multicolumn{2}{c}{All coordinates}\\
& \multicolumn{2}{c}{($n=20$)} & \multicolumn{2}{c}{($n=10$)} & \multicolumn{2}{c}{($n=10$)}\\
\cmidrule(lr){2-3}\cmidrule(lr){4-5}\cmidrule(lr){6-7}
$\Ntr$ & exact & plug-in & exact & plug-in & exact & plug-in\\
\midrule
$4n$ & $+0.022$ & $-0.010$ & $-0.004$ & $-0.029$ & $+0.105$ & $-0.068$ \\
$5n$ & $+0.024$ & $+0.002$ & $-0.006$ & $-0.018$ & $+0.079$ & $-0.046$ \\
$6n$ & $+0.021$ & $+0.009$ & $-0.005$ & $-0.008$ & $+0.065$ & $-0.035$ \\
$7n$ & $+0.021$ & $+0.017$ & $-0.006$ & $-0.010$ & $+0.063$ & $-0.018$ \\
$8n$ & $+0.030$ & $+0.021$ & $-0.005$ & $-0.006$ & $+0.062$ & $-0.006$ \\
$10n$ & $+0.027$ & $+0.024$ & $-0.001$ & $-0.001$ & $+0.042$ & $-0.009$ \\
$12n$ & $+0.028$ & $+0.025$ & $-0.002$ & $-0.001$ & $+0.043$ & $+0.002$ \\
$16n$ & $+0.029$ & $+0.030$ & $+0.004$ & $+0.006$ & $+0.044$ & $+0.014$ \\
$20n$ & $+0.030$ & $+0.031$ & $-0.001$ & $+0.001$ & $+0.045$ & $+0.023$ \\
$30n$ & $+0.032$ & $+0.033$ & $+0.009$ & $+0.011$ & $+0.036$ & $+0.022$ \\
known $F$ & $+0.037$ & $+0.037$ & $+0.006$ & $+0.006$ & $+0.027$ & $+0.027$ \\
\bottomrule
\end{tabular}
\end{table}
\FloatBarrier

\bibliographystyle{unsrtnat}
\bibliography{refs}

\end{document}